%% file: main.tex
\documentclass{article}

\usepackage[preprint]{neurips_2025} 

\input{math_commands.tex}

\usepackage{amsmath}
\usepackage{amssymb}
\usepackage{mathtools}
\usepackage{amsthm}
\usepackage{lipsum}
\usepackage{pifont}     
\usepackage{xcolor}     
\usepackage{bbm}
\usepackage{thm-restate} 

\usepackage[utf8]{inputenc}    
\usepackage[T1]{fontenc}       
\usepackage{microtype}         
\usepackage{xspace}            
\usepackage{enumitem}          

\usepackage{booktabs}          
\usepackage{tabularx}          
\usepackage{multirow}          
\usepackage{makecell}          
\usepackage{threeparttable}    
\usepackage{graphicx}          
\usepackage{subfigure}         
\usepackage{float}             
\usepackage{placeins}          
\usepackage{tablefootnote}
\usepackage{wrapfig}           
\usepackage{caption}
\usepackage{subcaption}

\usepackage{xcolor}            
\usepackage{colortbl}          
\usepackage[most]{tcolorbox}   

\usepackage[colorlinks=true,allcolors=perfblue]{hyperref}
\usepackage{url}               
\usepackage[capitalize]{cleveref}  

\usepackage[toc,page,header]{appendix}
\usepackage{minitoc}
\usepackage{titletoc}

\usepackage{algorithm}         
\usepackage{algorithmic}       

\theoremstyle{plain}
\newtheorem{theorem}{Theorem}[section]

\newtheorem{lemma}[theorem]{Lemma}

\theoremstyle{definition}

\theoremstyle{remark}

\newcommand{\modelname}{{\textsc{{MUSE}}}\xspace}

\definecolor{mydarkblue}{rgb}{0,0.08,0.45}
\definecolor{blue-color}{RGB}{102, 153, 204}
\definecolor{lightpink}{RGB}{255,230,230}
\definecolor{lightgreen}{RGB}{225,255,225}
\definecolor{lightblue}{RGB}{230,240,255}
\definecolor{lightpurple}{RGB}{217,217,255}
\definecolor{lightyellow}{RGB}{255,251,226}
\definecolor{perfblue}{RGB}{64, 114, 175}

\title{\modelname: Model-Agnostic Tabular Watermarking via Multi-Sample Selection}

\author{
Liancheng Fang\textsuperscript{1},
~~ Aiwei Liu\textsuperscript{2}\thanks{Corresponding author.},
~~ Henry Peng Zou\textsuperscript{1},
~~ Yankai Chen\textsuperscript{1}, \\
~~ \textbf{Hengrui Zhang\textsuperscript{1}},
~~ \textbf{Zhongfen Deng\textsuperscript{1}},
~~ \textbf{Philip S. Yu\textsuperscript{1}} \\
\textsuperscript{1}University of Illinois Chicago \quad
\textsuperscript{2}Tsinghua University \\
{\tt\small lfang87@uic.edu, liuaw20@mails.tsinghua.edu.cn, psyu@uic.edu}
}

\ifdefined\usebigfont

\usepackage{times}
\usepackage[fontsize=13pt]{scrextend}
\makeatletter
\@ifpackageloaded{geometry}{\AtBeginDocument{\newgeometry{letterpaper,left=1.56in,right=1.56in,top=1.71in,bottom=1.77in}}}{\usepackage[letterpaper,left=1.56in,right=1.56in,top=1.71in,bottom=1.77in]{geometry}}
\AtBeginDocument{\newgeometry{letterpaper,left=1.56in,right=1.56in,top=1.71in,bottom=1.77in}}
\usepackage{hyperref} 
\else
\fi

\begin{document}
\maketitle

\begin{abstract}
    We introduce \modelname, a watermarking algorithm for tabular generative models. Previous approaches typically leverage DDIM invertibility to watermark tabular diffusion models, but tabular diffusion models exhibit significantly poorer invertibility compared to other modalities, compromising performance. Simultaneously, tabular diffusion models require substantially less computation than other modalities, enabling a multi-sample selection approach to tabular generative model watermarking. \modelname embeds watermarks by generating multiple candidate samples and selecting one based on a specialized scoring function, without relying on model invertibility. Our theoretical analysis establishes the relationship between watermark detectability, candidate count, and dataset size, allowing precise calibration of watermarking strength. Extensive experiments demonstrate that \modelname achieves state-of-the-art watermark detectability and robustness against various attacks while maintaining data quality, and remains compatible with any tabular generative model supporting repeated sampling, effectively addressing key challenges in tabular data watermarking. Specifically, it reduces the distortion rates on fidelity metrics by $\mathbf{81-89}\%$, while achieving $\mathbf{1.0}$ TPR@0.1\%FPR detection rate. Implementation of \modelname can be found at \url{https://github.com/fangliancheng/MUSE}.

\end{abstract}

\input{Contents/1_introduction}

\input{Contents/3_preliminary}

\input{Contents/4_methods}

\input{Contents/5_experiments}

\input{Contents/2_related_work}
\input{Contents/6_conclusion}


\clearpage
\newpage
\bibliographystyle{plainnat}
\bibliography{ref}

\newpage
\appendix
\addcontentsline{toc}{section}{Appendix}
\section*{\LARGE Appendix}
\startcontents[part]
\printcontents[part]{}{1}{}
\newpage
\input{Contents/appendix}

\end{document}

%% file: math_commands.tex
\usepackage{amsmath,amsfonts,bm}
\usepackage{ulem}

\def\eqref#1{equation~\ref{#1}}

\def\1{\bm{1}}

\def\rc{{\textnormal{c}}}

\def\rw{{\textnormal{w}}}

\def\rvx{{\mathbf{x}}}

\def\rvz{{\mathbf{z}}}

\DeclareMathAlphabet{\mathsfit}{\encodingdefault}{\sfdefault}{m}{sl}
\SetMathAlphabet{\mathsfit}{bold}{\encodingdefault}{\sfdefault}{bx}{n}

\def\gJ{{\mathcal{J}}}

\def\gR{{\mathcal{R}}}



%% file: Contents/1_introduction.tex
\section{Introduction}
\label{sec:introduction}

\begin{wrapfigure}[17]{r}{0.5\textwidth}
\vspace{-19pt}
\centering
\includegraphics[width = 1.0\linewidth]{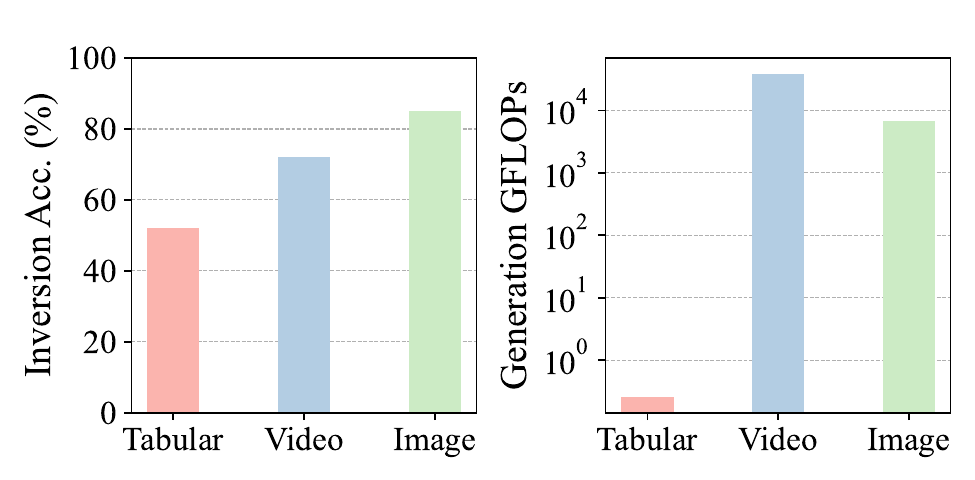}
\vspace{-15pt}
\caption{\textbf{Left}: Tabular diffusion models exhibit the lowest inversion accuracy (bit accuracy) when compared to video and image diffusion models. \textbf{Right}: Tabular diffusion models require much fewer generation GFLOPs than video and image diffusion models. Models used: TabSyn~\citep{tabsyn} (tabular), Stable Diffusion~\citep{sd,svd} (image/video).}
\label{fig:teaser}
\end{wrapfigure}

The rapid development of tabular generative models~\citep{tabddpm,tabmt,tbart,tabsyn,tabdiff,tabdar,tabgen_icl} has significantly advanced synthetic data generation capabilities for structured information. These breakthroughs have enabled the creation of high-quality synthetic tables for applications in privacy preservation, data augmentation, and missing value imputation~\citep{diffputer,hernandez2022synthetic, fonseca2023tabular, assefa2020generating}.
However, this advancement concurrently raises serious concerns about potential misuse, including data poisoning~\citep{padhi2021tabular} and financial fraud~\citep{cartella2021adversarial}. To address these risks, watermarking techniques have emerged as a pivotal technique. By embedding imperceptible yet robust signatures into synthetic data, watermarking facilitates traceability, ownership verification, and misuse detection~\citep{liu2024survey}.

Earlier works on tabular data watermarking~\citep{tabularmark, wgtd} utilize \textbf{\textit{edit-based watermarking}}, embedding signals by modifying table values. However, this approach has a fundamental limitation with tabular data: direct value alterations, especially in columns with discrete or categorical data, \textbf{can easily corrupt information} or render entries invalid. For instance, such edits might introduce non-existent categories or generate values outside permissible numerical ranges, significantly compromising data integrity.
Recently, \textbf{\textit{generative watermarking}} has emerged as an alternative approach for tabular data, drawing from successful techniques in diffusion models for images and videos~\citep{wu2025survey, gaussian_shading, tree_ring_watermark,hu2025videoshield}. This approach exploits the reversibility of DDIM samplers by setting patterned initial Gaussian noise and measuring its correlation with noise recovered through the inverse process. TabWak~\citep{tabwak} applies this concept to tabular diffusion models. Unlike edit-based watermarking, this method maintains better generation quality since the watermark is embedded within noise patterns that closely resemble Gaussian distributions, minimizing impact on the generated content.

However, watermarking tabular diffusion models is \textbf{significantly more challenging} than for image and video diffusion models. This stems from the \textbf{substantially lower accuracy of DDIM inverse processes} in tabular diffusion models, as shown in \Cref{fig:teaser} (left). When using the same Gaussian shading algorithm \citep{gaussian_shading}, tabular modality exhibits the lowest reversibility accuracy. This challenge arises because tabular diffusion models incorporate multiple additional algorithmic components that are difficult to reverse, such as quantile normalization \citep{wikipedia_quantile_normalization} and Variational Autoencoders (VAEs) \citep{vae} used in TabSyn\citep{tabsyn}. During watermark detection, the entire data processing pipeline must be inverted to recover the watermark signal, but this process accumulates errors as precisely reversing each step is often difficult or impossible. Key challenges in the inversion process include: (1) inverting quantile normalization is inherently problematic as this transformation is non-injective; (2) VAE decoder inversion relies on optimization methods without guarantees of perfect implementation. Due to limitations in tabular DDIM inversion accuracy, watermark detectability becomes highly dependent on model implementation, severely restricting its application scope and practical utility.

In this work, we introduce \modelname, a \textbf{model-agnostic} watermarking algorithm for tabular data that operates without relying on the invertibility of diffusion models.
A key insight enabling our approach is that tabular data generation demands \textbf{significantly less computation} than image or video generation, as shown in \Cref{fig:teaser} (right). This computational efficiency makes a multi-sample selection process practical.
\modelname leverages this by generating $m$ candidate samples for each data row. A watermark is embedded by selecting one candidate based on a score function, which is calculated using values from specific columns.
Because it does not depend on model invertibility, \modelname is broadly applicable to any tabular generative model. We also show that our method can maintain the original data distribution (distortion-free) when using certain column selection techniques. Furthermore, we establish a mathematical relationship between detectability, the number of candidates ($m$), and the number of rows ($N$), ensuring reliable watermark detection.

With these mechanisms, \modelname achieves high watermark detectability while preserving the generation quality of the underlying model, as validated across diverse tabular datasets and evaluation metrics. Moreover, \modelname exhibits strong robustness against a wide range of post-generation attacks specific to tabular data.
Importantly, \modelname is model-agnostic and compatible with any tabular generative model supporting repeated sampling. This flexibility allows \modelname to integrate seamlessly with sampling-efficient diffusion models, such as TabSyn~\citep{tabsyn}, thereby effectively mitigating the computational overhead introduced by repeated sampling processes.
We summarize the main contributions of this paper as follows:
\begin{itemize}[leftmargin=*]
\item We propose tabular watermarking via \underline{mu}lti-\underline{s}ample s\underline{e}lection (\modelname), a novel generative watermarking method for tabular data that completely avoids the inversion process in the data processing and sampling pipeline.
\item We provide a theoretical analysis of the detectability of \modelname, which enables precise calibration of watermarking strength under a given detection threshold.
\item Extensive experiments across multiple tabular datasets demonstrate that \modelname consistently achieves state-of-the-art performance in generation quality, watermark detectability, and robustness against various tabular-specific attacks.
\end{itemize}

%% file: Contents/3_preliminary.tex
\section{Preliminaries}
\label{sec:preliminary}



\paragraph{Tabular Generative Models.}
A tabular dataset with $N$ rows and $M$ columns consists of $i.i.d.$ samples $(\rvx_i)_{i=1}^N$ drawn from an unknown joint distribution $p(\mathbf{x})$, where each $\mathbf{x}_i \in \mathbb{R}^M$ (or mixed-type space) represents a data row with $M$ features. A tabular generative model aims to learn a parameterized distribution $p_{\theta}(\mathbf{x}) \approx p(\mathbf{x})$ to generate new realistic samples.

\paragraph{Watermark for Tabular Generative Models.}  Tabular watermark involves two main functions. 
\begin{enumerate}[leftmargin=2em,label=\arabic*.]
    \item \textbf{Generate}: Given a secret watermark key $k$, this function produces a watermarked table. Similar to standard generation, each row of this table is sampled i.i.d., but from a distribution $p(\mathbf{x}, k)$. 
    \item \textbf{Detect}: Provided with a table and a specific key $k$, this function examines the table to determine if it carries the watermark associated with that particular key.
\end{enumerate}


\paragraph{Threat Model.}
We consider the following watermarking protocol between three parties: the tabular data provider, the user, and the detector.
\begin{enumerate}[leftmargin=2em,label=\arabic*.]
    \item The tabular data provider shares a watermark key $k$ with the detector.
    \item The user asks the tabular data provider to generate a table $T$.
    \item The user publishes a table $T'$, which can either be an (edited version of the) original table $T$ or an independent table.
    \item The detector determines whether the table $T'$ is watermarked or not.
\end{enumerate}

%% file: Contents/4_methods.tex
\section{Proposed Method}
\label{sec:methods}

In this section, we introduce \modelname, a model-agnostic watermarking method for tabular generative models. Considering the low reverse accuracy of tabular diffusion models and fast generation compared to other modalities, our method is based on multi-sample selection instead of DDIM inversion.
We first introduce the overall watermark generation process in \Cref{sec:watermark_overview} and then introduce the Watermark Scoring Function and Column Selection Mechanism in  \Cref{sec:watermark_score} and \Cref{sec:adaptive_column_selection} respectively, followed by the detailed watermark detection in \Cref{sec:wm_detection}.

\begin{figure*}[!t]
    \centering
    \includegraphics[width=1.0\textwidth]{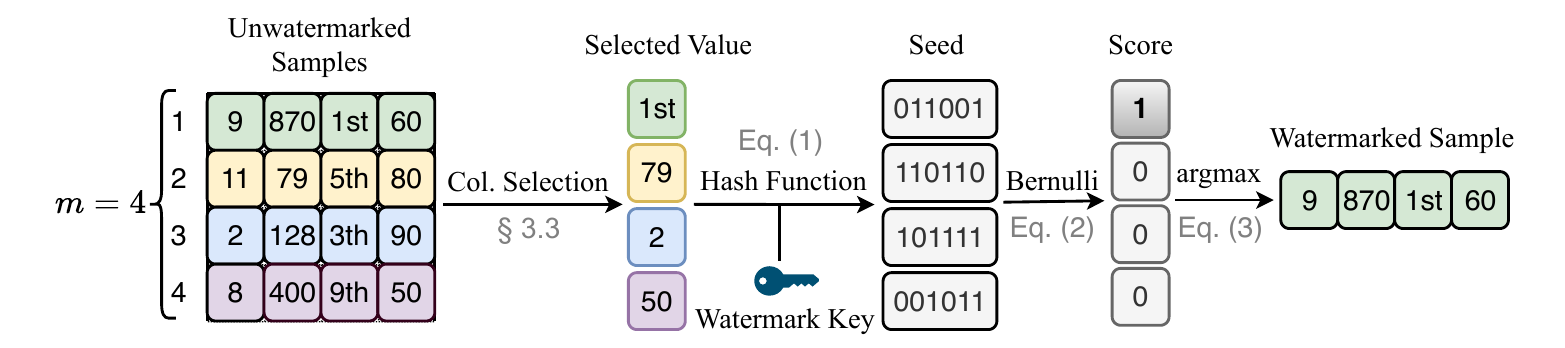}
    \caption{An overview of the \modelname watermark generation process. \modelname operates by generating multiple samples and selecting the highest-scoring sample (ties are broken randomly). The selected row is appended to the watermarked table, while others are discarded.}
    \vspace{-0.4cm}
    \label{fig:method_pipeline}
\end{figure*}

\subsection{Watermark Generation Framework}
\label{sec:watermark_overview}

We define the overall generation process of our \modelname method in this section. The generation of each watermarked row can be decomposed into the following two steps:

\fbox{%
   \begin{minipage}{\textwidth}
\begin{enumerate}[leftmargin=2em,label=\arabic*.,itemsep=0pt]
\item Sample $m$ candidate rows from the tabular generative model $p(\rvx)$.
\item Apply a \textit{watermark scoring function} $s_k(\rvx)$ to each candidate using watermark key $k$ and select the \textbf{highest-scoring candidate} as the watermarked row. Details of the watermark scoring function will be introduced in section \ref{sec:watermark_score}.
\end{enumerate}
\end{minipage}
}

Then we could repeat the above process $N$ times to generate the watermarked table.  In practice, the selection procedure can be parallelized across the $N$ groups since each group contains $i.i.d.$
samples. The overall process is illustrated in \Cref{fig:method_pipeline}.
\paragraph{Discussions.}
The advantages of this approach are twofold: (1) \textit{Model-agnostic}: since we do not modify the row generation process of the tabular generative model, the method is applicable to any tabular generative model that supports repeated sampling. (2) \textit{Distribution-preserving}: we demonstrate in subsequent sections that the watermark remains unbiased under specific watermark scoring functions.

\subsection{Watermark Scoring Function}
\label{sec:watermark_score}

In this section, we detail the design of our watermark scoring function, denoted as $s(\rvx, k)$, which involves two key steps as follows: 

\fbox{%
   \begin{minipage}{\dimexpr\textwidth-2\fboxsep-2\fboxrule\relax} 
\begin{enumerate}[leftmargin=2em,label=\arabic*.,itemsep=0pt]
    \item \textbf{Column Selection Mechanism}:
    Given an input sample $\rvx$ which consists of $M$ columns, the first step involves selecting a specific subset of $n$ columns from these $M$ available columns. Let $\pi(\rvx)$ represent the operation that extracts the values from this chosen subset of $n$ columns of $\rvx$. The specific strategies for how these $n$ columns are selected (e.g., adaptive or fixed methods) will be further elaborated in \Cref{sec:adaptive_column_selection}. 

    \item \textbf{Score Generation}:
    The second step takes the selected column values $\pi(\rvx)$ from the previous stage and the predefined watermark key $k$ as inputs. These are processed by a cryptographic hash function $H(\cdot)$ to produce a deterministic hash value:
    \begin{equation} \label{eq:random_seed_generation_revised}
        h = H(\pi(\rvx), k).
    \end{equation}
    This resulting hash value $h$ is then used to seed or as an input to a pseudorandom function $f$. This function $f$ maps the hash $h$ to a scalar score $s \in [0,1]$. In our implementation, we define $f$ as a Bernoulli distribution with a mean of $0.5$:
\begin{equation}  \label{eq:scoring_function}
    s_k(\rvx_i) = f(h_i), \quad f \sim \operatorname{Bernoulli}(0.5).
\end{equation}
    The rationale behind choosing a Bernoulli(0.5) distribution is discussed in \Cref{lem:optimal_scoring_distribution}.
\end{enumerate}
\end{minipage}
}

\paragraph{Discussion on Scoring Design.}
An important question is why our scoring function incorporates column selection rather than simply using $H(k)$ as the score. The key insight is that content-dependent scoring through $\pi(\rvx)$ makes each row's selection appear \textbf{pseudo-random} while still respecting the original data distribution $p(\rvx)$. If we used a fixed score based solely on the watermark key, the multi-sample selection process would become arbitrary, potentially biasing the resulting distribution away from $p(\rvx)$.

\subsection{Column Selection Mechanism}
\label{sec:adaptive_column_selection}

In this section, we introduce two strategies for the column selection procedure $\pi(\cdot)$ used in our watermark scoring function (\Cref{eq:random_seed_generation_revised}):

\fbox{%
   \begin{minipage}{\dimexpr\textwidth-2\fboxsep-2\fboxrule\relax}
\begin{enumerate}[leftmargin=2em,label=\arabic*.,itemsep=0pt]
    \item \textbf{Fixed Column Selection}: 
    A straightforward approach that uses a predefined set of columns for all samples, regardless of data characteristics.
    
    \item \textbf{Adaptive Column Selection}:
    A dynamic approach that selects columns based on properties of each individual sample, as detailed below.
\end{enumerate}
\end{minipage}
}

For the adaptive strategy, we leverage each sample's statistical properties. Unlike sequential language models that naturally provide randomness~\citep{kgw,synthid,unigram}, tabular generation is often non-sequential. Therefore, we use the sample itself as a source of entropy.

Specifically, we select columns based on how sample $\rvx$ deviates from the training data distribution $T$. For each column $j \in \{1, \dots, M\}$, we compute the empirical quantile rank:
\begin{equation} \label{eq:quantile_rank}
    r_j = \frac{1}{N} \sum_{i=1}^N \mathbbm{1}\{T_{i,j} \le \rvx_j\},
\end{equation}
where $N$ is the number of training samples.
We then select columns (typically 3) corresponding to the smallest, median, and largest $r_j$ values, resulting in an index set $\mathcal{J} = \pi(\rvx)$ for hash computation.

\paragraph{Discussion.} The adaptive column selection offers several advantages: (1) \textit{Enhanced Security}: By varying columns between samples, it prevents attackers from targeting specific columns; (2) \textit{Increased Diversity}: The diversity in column selection improves the uniformity of hash values, minimizing distortion in watermarked data; (3) \textit{Improved Robustness}: By linking selection to the data's empirical distribution, the method better withstands post-processing attacks that don't significantly alter the overall statistical properties. The complete \modelname watermark generation algorithm is presented in \Cref{alg:wm_generation}.

\paragraph{Repeated Column Masking.} While our column selection mechanism introduces diversity, achieving \textbf{true pseudo-randomness} requires an additional safeguard. We adopt a repeated column masking strategy that detects when previously selected column values reappear. In such cases, the scoring function assigns a random score, thereby skipping watermark embedding for that instance. This mechanism is inspired by the \textit{repeated key masking} technique used in LLM watermarking \citep{hu2023unbiased, synthid}. Further implementation details and experimental results are provided in Section~\ref{sec:ablation_study}.

\subsection{Watermark Detection}
\label{sec:wm_detection}
\paragraph{Detection Statistic.}
Since \modelname biases the selection process toward high-scoring samples, the sum of scores for a watermarked table is expected to be higher than that of an unwatermarked table. Given a (watermarked or unwatermarked) table $T$ consists of $N$ rows: $T := (\rvx_1, \ldots, \rvx_N)$. We detect watermark by computing the mean score:
\begin{equation} \label{eq:wm_detection_statistic}
    S(T) = \frac{1}{N} \sum_{i=1}^N s_k(\rvx_i).
\end{equation}
We then compare with the mean score $S(T_{\mathrm{no\text{-}wm}})$ of an unwatermarked table $T_{\mathrm{no\text{-}wm}}$. We conclude that $T$ is watermarked if $S(T) > S(T_{\mathrm{no\text{-}wm}})$.

\paragraph{Calibrating the Number of Repeated Samples.}
Given the detection statistic \Cref{eq:wm_detection_statistic}, we move on to show how the detectability of \modelname depends on 
(1) the number of watermarked samples $N$ and (2) the number of repeated samples $m$. 

\begin{restatable}{lemma}{falsePositiveRateBound}
    \label{lem:fpr_bound}
    Denote a watermarked table as $T_{\mathrm{wm}}$ and an unwatermarked table as $T_{\mathrm{no\text{-}wm}}$, each consisting of $N$ rows. Let $\rvx \sim p(\rvx)$ be a random variable drawn from the data distribution, and let $\rvx_1, \dots, \rvx_m$ be i.i.d. samples from $p(\rvx)$. Define $\mu_{\mathrm{no\text{-}wm}} = \mathbb{E}_{\rvx \sim p(\rvx)}[s_k(\rvx)]$ as the expected score of an unwatermarked sample, and define $\mu_{\mathrm{wm}}^{m} = \mathbb{E}_{\rvx_i \sim p(\rvx)}\left[ \max_{i \in [m]} s_k(\rvx_i) \right]$ as the expected score of a watermarked sample obtained via $m$ repeated samples. Suppose the scoring function satisfies $s_k(\cdot) \in [0,1]$, then, the False Positive Rate (FPR) of the watermark detection satisfies:
    \begin{equation} \label{eq:fpr_bound}
        \Pr\left(S(T_{\mathrm{no\text{-}wm}}) > S(T_{\mathrm{wm}})\right) \le 
        \exp\left( -\frac{N (\mu_{\mathrm{wm}}^{m} - \mu_{\mathrm{no\text{-}wm}})^2}{2} \right).
    \end{equation}
\end{restatable}

\begin{restatable}[Optimal Scoring Distribution]{lemma}{optimalScoringDistribution}
    \label{lem:optimal_scoring_distribution}
    Let $s_k(\rvx)$ be any random variable supported on $[0,1]$ with mean $0.5$, the right-hand-side of \Cref{eq:fpr_bound}
    is minimized when $s_k(\rvx)$ follows a $\text{Bernoulli}(0.5)$  distribution.
\end{restatable}

\begin{restatable}[Minimum Watermarking Signal]{theorem}{minimumWatermarkingSignal}
    \label{thm:optimal_number_of_repeated_samples}
    Under the same assumptions as in \Cref{lem:fpr_bound}, suppose the scoring function $s_k(\rvx)$ is instantiated as a hash-seeded pseudorandom function such that $s_k(\rvx) \sim \text{Bernoulli}(0.5)$, the FPR is upper-bounded by:
    \begin{equation} \label{eq:fpr_bound_optimal}   
        \Pr\left(S(T_{\mathrm{no\text{-}wm}}) > S(T_{\mathrm{wm}})\right) \le 
        \exp\left( -\frac{N}{2} \left(0.5 - 0.5^m\right)^2 \right).
    \end{equation}
    To ensure the FPR does not exceed a target threshold $\alpha$, it suffices to set the number of repeated samples $m$ as:
    \begin{equation} \label{eq:optimal_m}
        m = \max\left(2, \left\lceil \log_{0.5} \left(0.5 - \sqrt{\tfrac{2 \log(1/\alpha)}{N}} \right) \right\rceil\right),
    \end{equation}
    where $\lceil \cdot \rceil$ denotes the ceiling function. This expression is valid when $N > 8 \log(1/\alpha)$.
\end{restatable}

\begin{wrapfigure}[13]{r}{0.32\textwidth}
    \centering
    \includegraphics[width = 1.0\linewidth]{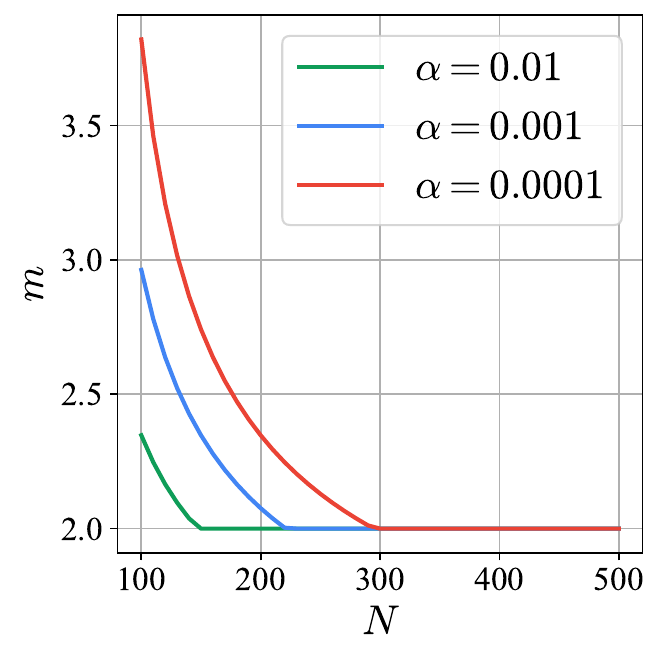}
    \vspace{-21pt}
    \caption{$m$ vs. $N$ under different $\alpha$ values (smoothed).}
    \label{fig:m_vs_N_alpha}
\end{wrapfigure}

\Cref{thm:optimal_number_of_repeated_samples} enables \modelname to calibrate the number of repeated samples $m$ to achieve a target false positive rate with theoretical guarantees. 
This allows the method to embed \textit{just enough} watermarking signal to ensure the desired detectability. Intuitively, since no redundant watermarking signal is embedded, the impact of watermarking on the generation quality is minimal. 
In \Cref{fig:m_vs_N_alpha}, we plot $m$ as a function of table size $N$ for various target FPRs, based on \Cref{eq:optimal_m} (omitting the ceiling operation for clarity). We observe that $m$ quickly saturates as $N$ increases. For instance, to achieve a 0.01\% FPR, $m=2$ suffices when $N \ge 300$, and even for $N=100$, $m=4$ is enough. In the rest of the paper, \modelname's $m$ is set by \Cref{eq:optimal_m} unless otherwise specified.

\begin{algorithm}[!t]
    \caption{\modelname Watermark Generation}
    \label{alg:wm_generation}
    \begin{algorithmic}[1]
        \STATE \textbf{Input:} watermark key $k$, unwatermarked table $T \in \mathbb{R}^{N \times M}$, False Positive Rate $\alpha$
        \STATE Compute the number of repeated samples $m$ based on $N$ and $\alpha$ via \Cref{eq:optimal_m}
        \STATE Randomly split rows of $T$ into $N/m$ groups: $(\mathcal{G}_i)_{i=1}^{N/m}$, each containing $m$ rows
        \STATE Initialize an empty set $\gR$, and a list $T_{wm}$ to store the watermarked table
        \FOR{$i \leftarrow 1$ to $N/m$}
            \STATE $\rvx_1, \dots, \rvx_m \leftarrow \mathcal{G}_i$
            \FOR[Adaptive column selection.]{$t \in \{1, \dots, m\}$}
                \STATE \textbf{Let} $\rvx \leftarrow \rvx_t$
                \STATE For each column $j$, compute quantile rank $r_j$ of $\rvx_j$ in $T[:,j]$ \COMMENT{See \Cref{eq:quantile_rank}.}
                \STATE Sort $\{r_j\}_{j=1}^M$ and identify column indices with min, median, and max ranks
                \STATE Let $\gJ_t$ be the set of selected column indices
            \ENDFOR
            \FOR[Multi-sample selection.]{$t \in \{1, \dots, m\}$} 
                \STATE $r_t = \operatorname{hash}(k, \rvx_t[\gJ_t])$
                \STATE Seed Bernoulli distribution with $r_t$
                \STATE $s_t \sim \operatorname{Bernoulli}(0.5)$
            \ENDFOR
            \STATE $i \leftarrow \operatorname{arg\,max}_{t \in \{1, \dots, m\}} s_t$ 
            \STATE Append $\rvx_i$ to $T_{wm}$
        \ENDFOR
        \STATE \textbf{return} $T_{wm}$
    \end{algorithmic}
\end{algorithm}

%% file: Contents/5_experiments.tex
\section{Experiments}
\label{sec:experiments}

\input{Tables/quality_detect.tex}

In this section, we provide a comprehensive empirical evaluation of \modelname. 
We aim to answer the following research questions:
\textbf{Q1: Detectability v.s. Invisibility} (\S\ref{sec:main_results}): Can \modelname achieve strong detectability while preserving the distribution of the generated data?
\textbf{Q2: Robustness} (\S\ref{sec:main_results}): How resilient is the watermark to a range of post-processing attacks, such as row/column deletion or value perturbation?
\textbf{Q3: Component-wise Analysis} (\S\ref{sec:ablation_study}): How does \modelname perform under different design choices of its components?

\subsection{Setup}
\label{sec:setup}
\paragraph{Datasets.}
We select six real-world tabular datasets containing both numerical and categorical attributes: \texttt{Adult}, \texttt{Default}, \texttt{Shoppers}, \texttt{Magic}, \texttt{Beijing} and \texttt{News}. Due to space constraints, we defer the results on \texttt{News} and \texttt{Beijing} to Appendix \ref{appendix:news_beijing_results}.
The statistics of the datasets are summarized in Table \ref{tbl:exp-dataset} in Appendix \ref{appendix:dataset}.

\paragraph{Evaluation Metrics.}
(a) To evaluate the detectability of the watermark, we report the area under the curve (AUC)
of the receiver operating characteristic (ROC) curve, and the True Positive Rate when the False
Positive Rate is at 0.1\%, denoted as \textit{TPR@0.1\%FPR}.
(b) To evaluate the distortion of the watermarked data, we follow standard fedelity and utility metrics used in tabular data generation \citep{tabsyn, tabddpm}: we report Marginal distribution (Marg.), Pair-wise column correlation (Corr.), Classifier-Two-Sample-Test (C2ST), and Machine Learning Efficiency (MLE). For MLE, we report the gap between the downstream task performance of the generated data and the real test set (MLE Gap). 
We refer the readers to \citep{tabsyn} for a more detailed definition of each evaluation metric.

\paragraph{Baselines.}
Since our method belongs to the class of generative watermarking techniques, we primarily compare it with TabWak~\citep{tabwak}, which is the only existing generative watermarking approach for tabular data. We use the official implementations of both TabWak and its improved variant TabWak* in all experiments to ensure consistency. In addition, following the experimental setup in TabWak, we include two image watermarking methods, TreeRing (TR) and Gaussian Shading (GS), as auxiliary baselines. For completeness, we also evaluate against representative edit-based watermarking methods, including TabularMark~\citep{tabularmark} and WGTD~\citep{wgtd}. Due to space constraints, detailed descriptions and results for these methods are deferred to \Cref{appendix:edit-based-results}.

\paragraph{Implementation Details.}
We use the same tabular generative model as TabWak, namely TabSyn~\citep{tabsyn}, and train it using the official codebase. For TabWak, we use its official implementation to generate watermarked data. Notably, it bypasses the inversion of quantile normalization, which assumes access to ground-truth data not available for watermark detection, potentially giving it an advantage under our evaluation protocol.
Generation quality is evaluated across ten repetitions, and we report the averaged results.

\subsection{Main Results}
\label{sec:main_results}
\paragraph{Distortion and Detectability.}
We address the first question: whether the watermarking method achieves high watermark detectability while introducing minimal distortion to the generated data. As shown in \Cref{tbl:quality_detect} and \Cref{appendix:news_beijing_results}, \modelname consistently achieves strong performance across both fidelity and detection metrics on all six datasets. It yields the highest marginal statistics, correlation, C2ST and MLE, often closely matching the unwatermarked baseline. This suggests that the underlying data distribution is well preserved.
In terms of detectability, \modelname achieves nearly perfect detection performance across all datasets and detection budgets, as measured by both AUC and T@0.1\%F. For example, on the \texttt{Default} dataset, it attains a marginal statistic of 0.983 and AUC of 1.000, compared to 0.911 and 0.896 for TabWak. Notably, while GS also achieves strong detection scores, this comes at the cost of significantly higher distortion across all fidelity metrics. For instance, on the \texttt{Adult} dataset, GS results in a C2ST of only 0.058 and correlation of 0.619, in contrast to the higher values of 0.883 and 0.963 from \modelname, respectively. These results indicate that \modelname embeds a detectable watermark signal while preserving the statistical properties of the generated data to a greater extent than existing approaches.

\paragraph{Robustness against Attacks.}
We assess the robustness of watermarking methods under five representative attacks on tabular data: \textit{row shuffling}, \textit{row deletion}, \textit{column deletion}, \textit{cell deletion}, and \textit{value alteration}. Each attack is applied at varying intensities, with attack percentages ranging from 0.0 to 1.0 in increments of 0.2. For deletion-based attacks, a specified fraction of rows, columns, or cells is randomly removed and replaced with unwatermarked values independently sampled from the same generative model. In the value alteration attack, selected numerical values are perturbed by multiplying each with a scalar drawn uniformly from $(0.8, 1.2)$. Row shuffling permutes a fraction of the dataset’s rows.
We benchmark the detectability of \modelname against TabWak and TabWak* on the \texttt{Adult} dataset, using $N{=}500$ and $m{=}2$. As shown in \Cref{fig:robustness}, \modelname consistently outperforms or matches the performance of TabWak and TabWak* under four of the five attacks—row shuffling, row deletion, cell deletion, and value alteration. Under column deletion, however, \modelname shows a drop in performance due to its watermark being embedded via selected columns, which are partially removed by this attack. For benchmark results on other datasets, we refer readers to Appendix~\ref{appendix:robustness_results}.

\begin{figure}[!t]
    \centering
    \includegraphics[width=\textwidth]{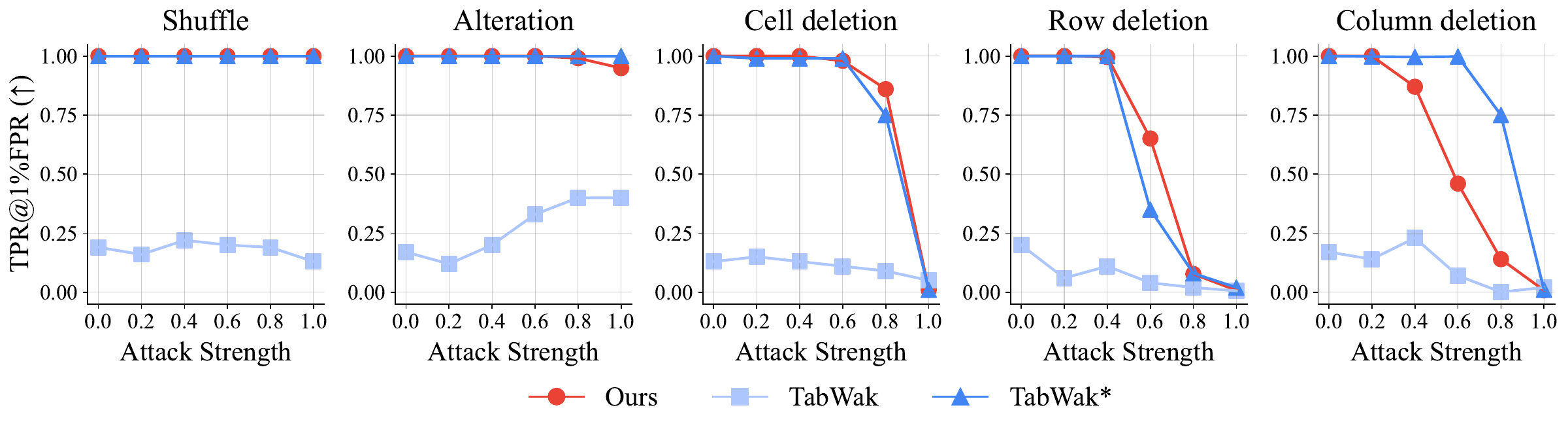}
    \vspace{-0.6cm}
    \caption{Detection performance of watermarking methods against different types of tabular data attacks across varying attack intensities.}
    \vspace{-0.6cm}
    \label{fig:robustness}
\end{figure}

\vspace{-0.3cm}
\input{Tables/ablation.tex}
\vspace{-0.3cm}

\subsection{Ablation Study and Further Analysis}  
\label{sec:ablation_study}

We perform a component-wise ablation study to evaluate the contribution of each design choice in our watermarking framework.  All experiments are conducted on the \texttt{Adult} dataset and we generate watermarked table with $N=100$ rows, if not otherwise specified. For detectability, we report the $z$-statistic defined as $\frac{\sum_{i=1}^{N} s_k(\rvx_i) - N/2}{\sqrt{N/4}}$. 
\vspace{-0.3cm}
\paragraph{Score Function.}
We compare two scoring distributions: (1) a Bernoulli distribution with mean $0.5$, and (2) a uniform distribution over $[0,1]$. As shown in \Cref{tbl:ablation}, \modelname with the Bernoulli score yields higher detectability. This result aligns with our theoretical analysis in \Cref{lem:optimal_scoring_distribution}, which identifies Bernoulli$(0.5)$ as the optimal scoring distribution under our detection formulation.

\vspace{-0.2cm}
\paragraph{Column Selection.}
We compare adaptive column selection strategy with a fixed set strategy that selects the first three columns. As shown in \Cref{tbl:ablation}, adaptive column selection leads to higher detectability.
We also investigate the effect of varying the number of selected columns. Increasing the number of selected columns generally improves both detectability and generation quality due to improved diversity for the hash function. However, using more columns increases vulnerability to column deletion attacks. Additionally, detectability tends to saturate beyond three columns. Thus, we use three adaptively selected columns in all main experiments for a balanced trade-off.

\paragraph{Distortion-Free Watermarking.}
Ideally, the selection process described in \Cref{sec:methods} would introduce no distortion to the data distribution if it were entirely independent of sample values. In practice, however, some dependence is necessary to ensure watermark detectability. To approximate this ideal, we leverage the insight that if a value in the selected columns has not been previously used for watermarking, its selection can be considered effectively random.
To enforce this, we implement a masking mechanism that tracks previously watermarked values and skips watermarking on samples that would reuse them. As shown in \Cref{tbl:ablation}, this mechanism helps preserve the underlying data distribution and improves generation quality. However, it also reduces the number of watermarked samples, slightly compromising overall detection strength.

\paragraph{Model-Agnostic Applicability.}
While our primary experiments are based on a diffusion model~\citep{tabsyn}, \modelname is universally applicable to any generative model and pre-allocated tabular data. To demonstrate this, we evaluate \modelname on two additional representative paradigms of tabular generative modeling:
(1) \textit{Autoregressive models}: we adopt DP-TBART~\citep{tbart}, a transformer-based autoregressive model that predicts each tabular entry conditioned on preceding entries; and  
(2) \textit{Masked generative models}: we use TabDAR~\citep{tabdar}, a masked autoencoder model that predicts randomly masked values.
As shown in \Cref{tbl:ablation}, \modelname consistently achieves strong detectability and generation quality across all three model families, confirming its generality and robustness across diverse generative architectures.


\paragraph{Computation Time.}
\begin{wrapfigure}{r}{0.48\textwidth}
    \centering
    \vspace{-1.2em}
    \includegraphics[width=\linewidth]{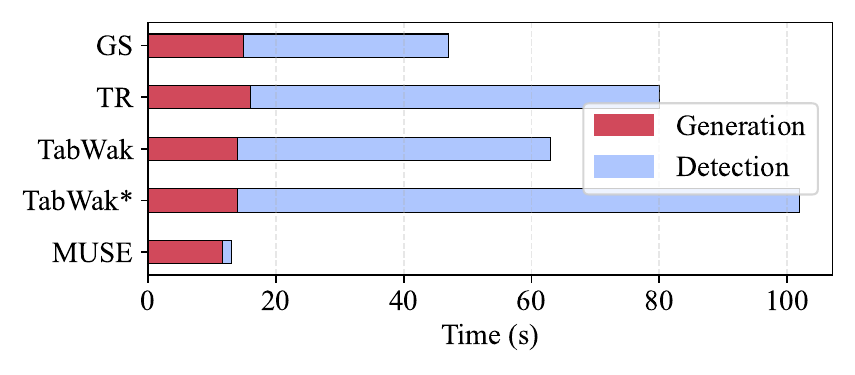}
    \vspace{-2.0em}
    \caption{Watermark generation and detection time of \modelname and inversion-based baselines.}
    \label{fig:computation}
    \vspace{-1em}
\end{wrapfigure}

We compare the effective watermarking time (generation + detection) of \modelname with baselines that rely on DDIM inversion. We generate 10K watermarked rows of the \texttt{Adult} dataset. As shown in \Cref{fig:computation}, \modelname achieves significantly lower detection time by avoiding the costly inversion process. Notably, its generation time is also lower than that of the baselines, despite using multi-sample generation ($m=2$). This efficiency arises from \modelname's compatibility with fast score-based diffusion models~\citep{tabsyn,edm}, which require only 50 sampling steps, compared to the 1,000 steps typically needed for DDIM inversion to ensure sufficient discretization.

%% file: Tables/quality_detect.tex
\begin{table}[t]
    \caption{Watermark generation quality and detectability, \raisebox{0.5ex}{\colorbox{blue!30}{\quad}} indicates best performance, \raisebox{0.5ex}{\colorbox{blue!10}{\quad}} indicates second-best performance. For clarity, only our method is highlighted in detection.}
    \vspace{8pt}
    \centering
    \resizebox{0.98\textwidth}{!}{
    \begin{threeparttable}
    \begin{tabular}{llcccc|cccc}
    \toprule[1.0pt]
    & &\multicolumn{4}{c}{Watermark Generation Quality} & \multicolumn{4}{c}{Watermark Detectability}\\
    \cmidrule(r){3-6} \cmidrule(r){7-10}
    Dataset & Method & \multicolumn{4}{c}{Num. Training Rows} &  \multicolumn{2}{c}{100} & \multicolumn{2}{c}{500} \\
    \cmidrule(r){3-6} \cmidrule(r){7-8} \cmidrule(r){9-10}
     & & \textbf{Marg.} & \textbf{Corr.} & \textbf{C2ST} & \textbf{MLE Gap} & \textbf{AUC} & \textbf{T@0.1\%F} & \textbf{AUC} & \textbf{T@0.1\%F} \\
    \midrule
    \multirow{7}{*}{Adult} & w/o WM & 0.994 & 0.984 & 0.996 & 0.017 & - & - & - & - \\ 
    & TR & 0.919 & 0.870 & 0.676 & \cellcolor{blue!10}0.046 & 0.590 & 0.004 & 0.774 & 0.171 \\
    & GS & 0.751 & 0.619 & 0.058 & 0.084 & 1.000 & 1.000 & 1.000 & 1.000 \\
    & TabWak & \cellcolor{blue!10}0.935 & \cellcolor{blue!10}0.885 & \cellcolor{blue!10}0.769 & 0.048 & 0.844 & 0.089 & 0.990 & 0.592 \\
    & TabWak* & 0.933 & 0.879 & 0.713 & 0.085 & 0.999 & 0.942 & 1.000 & 1.000 \\
    & \textbf{\modelname} & \cellcolor{blue!30}0.979 & \cellcolor{blue!30}0.963 & \cellcolor{blue!30}0.883 & \cellcolor{blue!30}0.017 & \cellcolor{blue!30}1.000 & \cellcolor{blue!30}1.000 & \cellcolor{blue!30}1.000 & \cellcolor{blue!30}1.000 \\
    \midrule
    \multirow{7}{*}{Default} & w/o WM & 0.990 & 0.934 & 0.979 & 0.000 & - & - & - & - \\ 
    & TR & 0.895 & 0.888 & 0.564 & 0.161 & 0.579 & 0.001 & 0.848 & 0.034 \\
    & GS & 0.701 & 0.678 & 0.059 & 0.182 & 1.000 & 1.000 & 1.000 & 1.000 \\
    & TabWak & \cellcolor{blue!10}0.911 & \cellcolor{blue!10}0.902 & \cellcolor{blue!10}0.568 & \cellcolor{blue!10}0.156 & 0.896 & 0.071 & 0.997 & 0.611 \\
    & TabWak* & 0.906 & 0.894 & 0.550 & 0.176 & 0.965 & 0.218 & 1.000 & 0.995 \\
    & \textbf{\modelname} & \cellcolor{blue!30}0.983 & \cellcolor{blue!30}0.925 & \cellcolor{blue!30}0.963 & \cellcolor{blue!30}0.002 & \cellcolor{blue!30}1.000 & \cellcolor{blue!30}1.000 & \cellcolor{blue!30}1.000 & \cellcolor{blue!30}1.000 \\
    \midrule
    \multirow{7}{*}{Magic} & w/o WM & 0.990 & 0.980 & 0.998 & 0.008 & - & - & - & - \\
    & TR & 0.898 & \cellcolor{blue!10}0.936 & \cellcolor{blue!10}0.621 & 0.129 & 0.652 & 0.014 & 0.592 & 0.102 \\
    & GS & 0.688 & 0.838 & 0.030 & \cellcolor{blue!10}0.064 & 1.000 & 1.000 & 1.000 & 1.000 \\
    & TabWak & \cellcolor{blue!10}0.905 & 0.929 & 0.605 & 0.120 & 0.904 & 0.067 & 0.997 & 0.737 \\
    & TabWak* & 0.891 & 0.916 & 0.520 & 0.100 & 0.873 & 0.050 & 0.995 & 0.687 \\
    & \textbf{\modelname} & \cellcolor{blue!30}0.991 & \cellcolor{blue!30}0.982 & \cellcolor{blue!30}0.999 & \cellcolor{blue!30}0.010 & \cellcolor{blue!30}1.000 & \cellcolor{blue!30}1.000 & \cellcolor{blue!30}1.000 & \cellcolor{blue!30}1.000 \\
    \midrule
    \multirow{7}{*}{Shoppers} & w/o WM & 0.985 & 0.974 & 0.974 & 0.017 & - & - & - & - \\
    & TR & 0.888 & 0.880 & 0.501 & \cellcolor{blue!10}0.077 & 0.575 & 0.001 & 0.830 & 0.058 \\
    & GS & 0.729 & 0.688 & 0.061 & 0.154 & 1.000 & 1.000 & 1.000 & 1.000 \\
    & TabWak & \cellcolor{blue!10}0.903 & \cellcolor{blue!10}0.886 & \cellcolor{blue!10}0.548 & 0.132 & 0.860 & 0.106 & 0.990 & 0.353 \\
    & TabWak* & 0.897 & 0.879 & 0.525 & 0.384 & 0.742 & 0.002 & 0.981 & 0.185 \\
    & \textbf{\modelname} & \cellcolor{blue!30}0.982 & \cellcolor{blue!30}0.974 & \cellcolor{blue!30}0.950 & \cellcolor{blue!30}0.015 & \cellcolor{blue!30}1.000 & \cellcolor{blue!30}1.000 & \cellcolor{blue!30}1.000 & \cellcolor{blue!30}1.000 \\
    \bottomrule[1.0pt]
    \end{tabular}
    \end{threeparttable}
    }
    \vspace{-0.3cm}
    \label{tbl:quality_detect}
\end{table}

%% file: Tables/ablation.tex
\begin{table}[htbp]
    \centering
    \small
    \caption{Component-wise ablation study of \modelname. Each color block indicates a different component of the method. Details of the experimental setup are in \S\ref{sec:ablation_study}.}
    \label{tbl:ablation}
    \resizebox{0.98\textwidth}{!}{%
    \begin{tabular}{clcccccccc}
    \toprule
    \textbf{Model} & \textbf{Score func.} & \textbf{Col. Select} & \textbf{Mask} & \textbf{Num. Col.} & $z$-\textbf{stat.$\uparrow$} & \textbf{Marg.$\uparrow$} & \textbf{Corr.$\uparrow$} & \textbf{C2ST$\uparrow$} & \textbf{MLE Gap$\downarrow$} \\
    \midrule
    \cellcolor{lightpink}TabSyn & Bernoulli & Adaptive & No & 3 & 7.348 & \textbf{0.979} & \textbf{0.963} & \textbf{0.883} & \textbf{0.017} \\
    \cellcolor{lightpink}TabDAR & Bernoulli & Adaptive & No & 3 & 7.270 & 0.977 & 0.958 & 0.880 & 0.018 \\
    \cellcolor{lightpink}DP-TBART & Bernoulli & Adaptive & No & 3 & \textbf{7.544} & 0.951 & 0.931 & 0.759 & 0.020 \\
    \midrule
    TabSyn & \cellcolor{lightgreen}Bernoulli & Adaptive & No & 3 & \textbf{7.348} & \textbf{0.979} & \textbf{0.963} & \textbf{0.883} & 0.017 \\
    TabSyn & \cellcolor{lightgreen}Uniform & Adaptive & No & 3 & 5.012 & 0.964 & 0.940 & 0.808 & \textbf{0.015} \\
    \midrule
    TabSyn & Bernoulli & \cellcolor{lightblue}Adaptive & No & 3 & \textbf{7.348} & \textbf{0.979} & \textbf{0.963} & \textbf{0.883} & 0.017 \\
    TabSyn & Bernoulli & \cellcolor{lightblue}Fixed & No & 3 & 5.439 & 0.949 & 0.907 & 0.601 & \textbf{0.015} \\
    \midrule
    TabSyn & Bernoulli & Adaptive & \cellcolor{lightyellow}No & 3 & \textbf{7.348} & 0.979 & 0.963 & 0.883 & \textbf{0.017} \\
    TabSyn & Bernoulli & Adaptive & \cellcolor{lightyellow}Yes & 3 & 4.819 & \textbf{0.985} & \textbf{0.973} & \textbf{0.940} & \textbf{0.017} \\
    \midrule 
    TabSyn & Bernoulli & Adaptive & No & \cellcolor{lightpurple}1 & 4.987 & 0.931 & 0.879 & 0.544 & \textbf{0.015} \\
    TabSyn & Bernoulli & Adaptive & No & \cellcolor{lightpurple}3 & 7.348 & 0.979 & 0.963 & 0.883 & 0.017 \\
    TabSyn & Bernoulli & Adaptive & No & \cellcolor{lightpurple}5 & 8.624 & 0.989 & 0.969 & 0.983 & 0.017 \\
    TabSyn & Bernoulli & Adaptive & No & \cellcolor{lightpurple}7 & \textbf{8.728} & \textbf{0.990} & \textbf{0.976} & \textbf{0.995} & 0.018 \\
    \bottomrule
    \end{tabular}
    } 
\end{table}

%% file: Contents/2_related_work.tex
\section{Related Work}
\label{sec:related_work}
\vspace{-0.1cm}

\paragraph{Generative Watermarking.}
Generative watermarking embeds watermark signals during the generation process, typically by manipulating the generation randomness through pseudorandom seeds. This approach has proven effective and efficient for watermarking in image, video, and large language model (LLM) generation.
In image and video generation, where diffusion-based models are the \textit{de facto} standard, watermarking methods inject structured signals into the noise vector in latent space~\citep{tree_ring_watermark,gaussian_shading,huang2024robin}. Detection involves inverting the diffusion sampling process~\citep{dhariwal2021diffusion,hong2024exact,pan2023effective} to recover the original noise vector and verify the presence of the embedded watermark.
For LLMs, generative watermarking methods fall into two categories: 
(1) \textit{Watermarking during logits generation}, which embeds signals by manipulating the model's output logits distribution~\citep{kgw,unigram,hu2023unbiased,synthid,watermax,Liu2023ASI}; and 
(2) \textit{Watermarking during token sampling}, which preserves the logits distribution but replaces the stochastic token sampling process (e.g., multinomial sampling) with a pseudorandom procedure seeded for watermarking~\citep{aaronson2022watermarking,kuditipudi2023robust,christ2024undetectable}. In this sense, sampling-based watermarking is conceptually similar to inversion-based watermarking used in diffusion models. We refer the reader to \citep{liu2024survey,markllm} for a comprehensive survey of watermarking for LLMs.
Closest to our approach are SynthID~\citep{synthid} and Watermax~\citep{watermax}, both of which embed watermarks via repeated logit generation. However, our approach is specifically designed for unconditional tabular data generation, unlike these methods which primarily target discrete text. This focus on tabular data introduces unique challenges due to its distinct data structure. Consequently, our watermarking technique is engineered for robustness against a different set of attacks prevalent in the tabular domain.

\paragraph{Watermarking for Tabular Data}
Traditional tabular watermarking techniques are edit-based, injecting signals by modifying existing data values. WGTD~\citep{wgtd} embeds watermarks by altering the fractional parts of continuous values using a green list of intervals, but it is inapplicable to categorical-only data. TabularMark~\citep{tabularmark} perturbs values in a selected numerical column using pseudorandom domain partitioning, but relies on access to the original table for detection, limiting its robustness in adversarial settings.
Another significant drawback of such methods is the potential to distort the original data distribution or violate inherent constraints.
To overcome this, TabWak~\citep{tabwak} introduced the first generative watermarking approach for tabular data. Analogous to inversion-based watermarks in diffusion models, TabWak embeds detectable patterns into the noise vector within the latent space. It also employs a self-clone and shuffling technique to minimize distortion to the data distribution.
While TabWak avoids post-hoc editing, its reliance on inverting both the sampling process (e.g., DDIM~\citep{song2020denoising}) and preprocessing steps (e.g., quantile normalization~\citep{wikipedia_quantile_normalization}) can introduce reconstruction errors. These errors will in turn impair the watermark's detectability.

\vspace{-0.2cm}

%% file: Contents/6_conclusion.tex
\section{Conclusion}
\label{sec:conclusion}

We propose \modelname, a model-agnostic watermarking method that embeds signals via multi-sample selection, eliminating the need for costly and error-prone inversion procedures. \modelname achieves strong detectability while introducing negligible distributional distortion and seamlessly scales across a wide range of generative models. Extensive experiments on benchmark datasets validate its effectiveness, consistently outperforming prior edit-based and inversion-based approaches in both generation fidelity and watermark robustness.
As synthetic tabular data becomes increasingly adopted in high-stakes domains such as healthcare, finance, and social science, ensuring data traceability and integrity is critical. \modelname provides a practical and generalizable solution for watermarking synthetic data, enabling reliable provenance tracking, ownership verification, and misuse detection. We believe this work opens new avenues for trustworthy synthetic data generation and highlights the importance of integrating security considerations into the core of data-centric AI systems.

%% file: Contents/appendix.tex
\section{Additional Experiments Results}

\subsection{Omitted Results on Distortion and Detectability}
\label{appendix:news_beijing_results}
We present the omitted results on distortion and detectability in Table~\ref{tbl:quality_detect_appendix}.
\input{Tables/quality_detect_appendix}

\subsection{Omitted Results on Robustness}
\label{appendix:robustness_results}
We present the omitted robustness results in \Cref{fig:all_comparisons}, where \modelname is compared against TabWak and TabWak* on the \texttt{Beijing, Default, Magic, News}, and \texttt{Shoppers} datasets. Overall, \modelname demonstrates stronger robustness under cell deletion and row deletion attacks, while achieving comparable performance on alteration and column deletion attacks. Both \modelname and TabWak/TabWak* remain resilient to shuffle attacks, due to embedding watermarks at the individual row level. Notably, we observe that TabWak and TabWak* exhibit instability on certain datasets, such as \texttt{Shoppers} and \texttt{Beijing}, where detection performance fluctuates—first decreasing and then increasing—as attack intensity increases. We hypothesize that this behavior stems from the inherent instability of the VAE inversion process.

\begin{figure}[t]
  \centering
  \includegraphics[width=\linewidth]{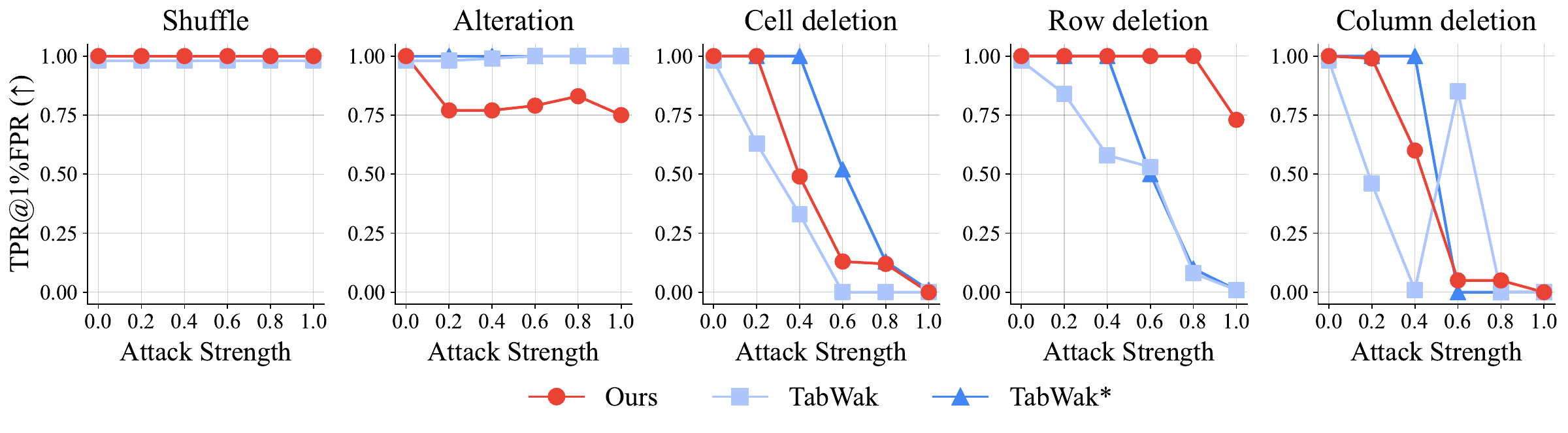}
  \vspace{0.5em}
  \includegraphics[width=\linewidth]{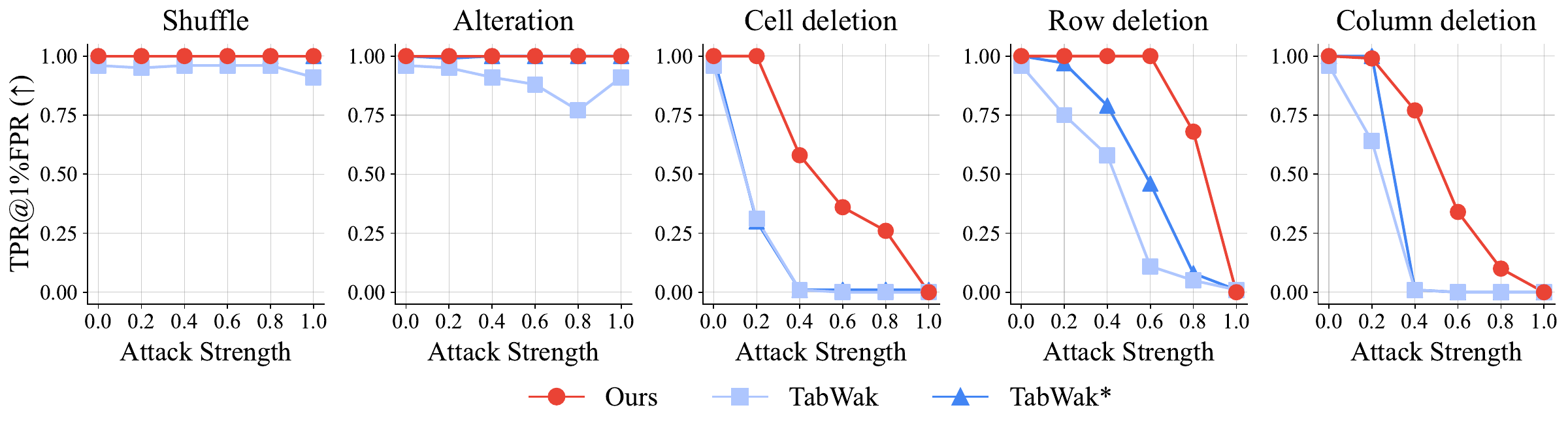}
  \vspace{0.5em}
  \includegraphics[width=\linewidth]{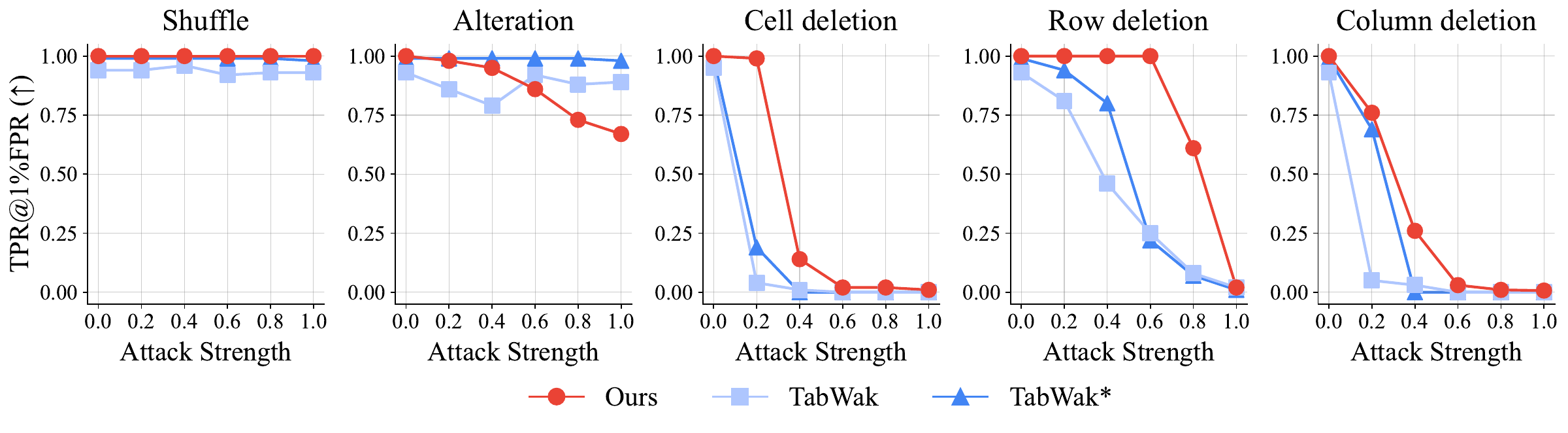}
  \vspace{0.5em}
  \includegraphics[width=\linewidth]{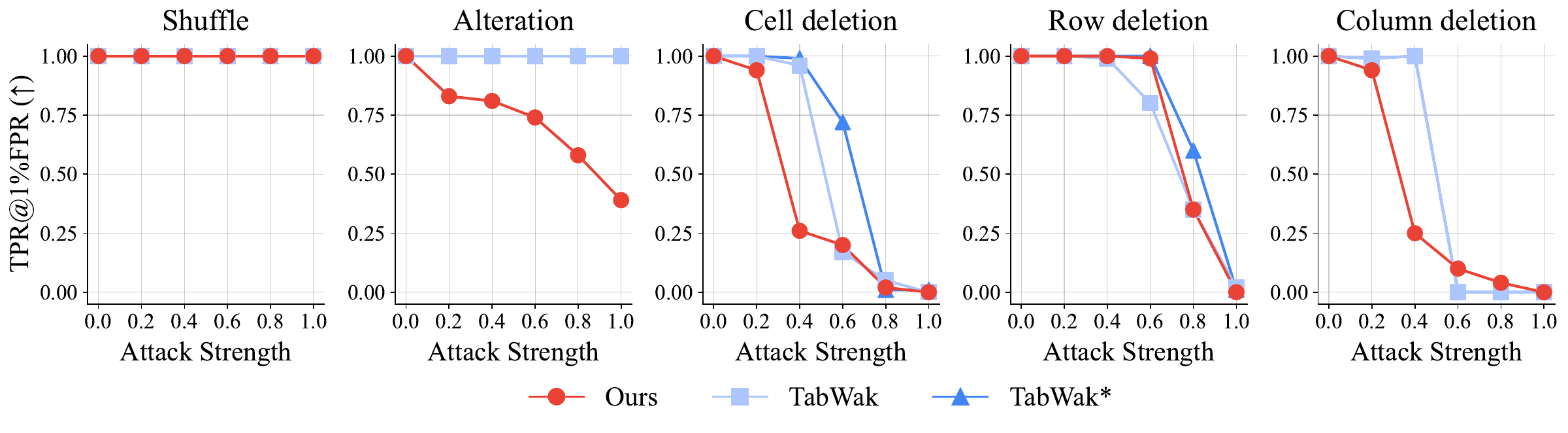}
  \vspace{0.5em}
  \includegraphics[width=\linewidth]{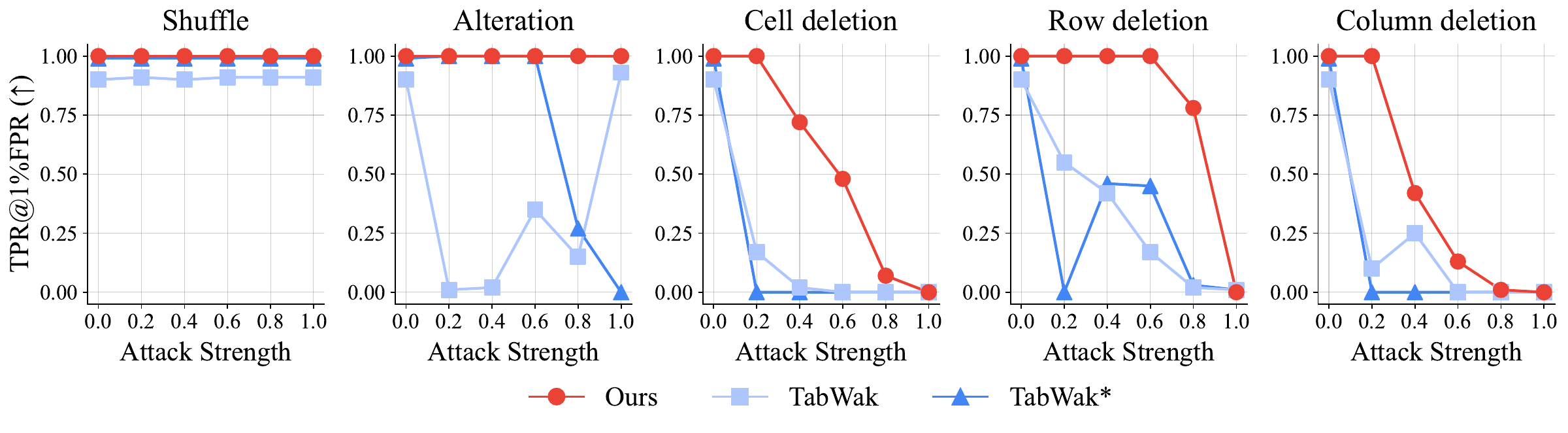}
  \caption{Detection performance of \modelname vs. TabWak/TabWak* against different types of tabular data attacks across varying attack intensities. From top to bottom: \texttt{Beijing, Default, Magic, News and Shoppers.}}
  \label{fig:all_comparisons}
\end{figure}

\subsection{Omitted Results on Edit-based Watermarking}
\label{appendix:edit-based-results}
\input{Tables/edit_baseline}

We compare our method against two representative \textbf{edit-based} watermarking baselines, which embed watermarks by directly altering table entries. Since the official implementations of these methods are not publicly available, we reimplement them based on the descriptions in their original papers. We first outline their core methodologies and our reimplementation details, then present the comparative results in Table~\ref{tbl:edit_baseline}. \textbf{Our reproduced codes are provided in the supplementary material.} Below are the
detailed implementations of the baselines.

\label{appendix:implementation}
\noindent\textbf{WGTD}~\citep{wgtd}.  
WGTD embeds watermarks by modifying the fractional part of continuous data points, replacing them with values from a predefined green list. Consequently, \textbf{it is limited to continuous data and cannot be applied to tables containing only categorical features}.

The watermarking process in WGTD involves three main steps: (i) dividing the interval $[0, 1]$ into $2m$ equal sub-intervals to form $m$ pairs of consecutive intervals; (ii) randomly selecting one interval from each pair to construct a set of $m$ ``green list'' intervals; and (iii) replacing the fractional part of each data point with a value sampled from the nearest green list interval, if the original does not already fall within one. Detection is performed via a hypothesis-testing framework that exploits the statistical properties of the modified distribution to reliably identify the presence of a watermark. For reproducibility, we adopt the original hyperparameter setting with $m = 5$ green list intervals.

\noindent\textbf{TabularMark}~\citep{tabularmark}.  
TabularMark embeds watermarks by perturbing specific cells in the data. It first pick a selected attribute/column to embed the watermark, then it generate pesudorandom partition of a fixed range into multiple unit domains, and label them with red and green domains, and finally perturb the selected column with a random number from the green domain.
In our implementation, we choose the first numerical column as the selected attribute, and set the number of unit domains $k=500$, the perturbation range controlled by $p=25$, and configure $n_w$ as 10\% of the total number of rows. 

During detection, TabularMark leverages the original unwatermarked table to reverse the perturbations and verify whether the restored differences fall within the green domain. However, \textbf{this approach assumes access to the original unwatermarked table}, which is often impractical, especially in scenarios where the watermarked table can be modified by adversaries.

\paragraph{Discussions.} 
As demonstrated in \Cref{tbl:edit_baseline}, both WGTD and TabularMark exhibit strong detection performance across all datasets. Furthermore, their generation quality is generally comparable to that of \modelname. However, a notable observation is the significant performance degradation measured by the MLE metric for both WGTD and TabularMark on the \texttt{Beijing} and \texttt{News} datasets. We hypothesize that this performance drop stems from the post-editing process, which may introduce substantial artifacts into the data. These artifacts, in turn, could negatively impact the performance of downstream machine learning tasks.

\section{Further Analysis of the Inversion-Based Watermarking}
\label{appendix:inversion-based-watermarking}
We first introduce the overall pipeline of inversion-based watermarking in \Cref{fig:tabsyn_architecture}. The difficulty lies in the inversion of three components, in sequential order: (1) inverse Quantile Transformation (IQT) \S\ref{appendix:iq_inversion}, (2) the VAE decoder \S\ref{appendix:vae_inversion}, and (3) the DDIM sampling process \S\ref{appendix:ddim_inversion}. Finally, we analyze the error accumulation and detection performance across the inversion stages in \S\ref{appendix:error_accumulation}.

\subsection{Pipeline of Inversion-based Watermarking}
\begin{figure}[H]
  \centering
  \includegraphics[width=\linewidth]{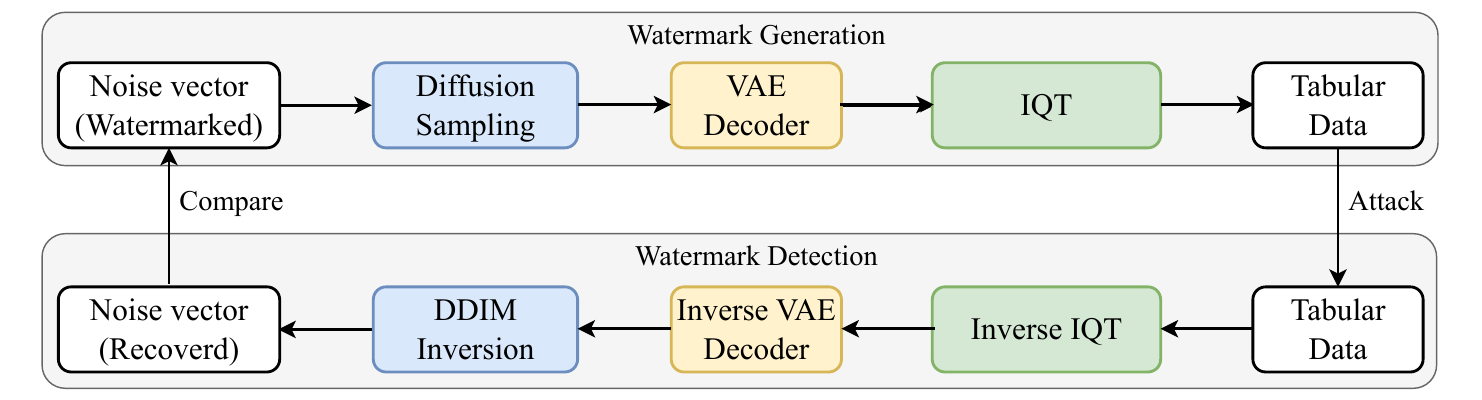}
  \caption{Pipeline of Inversion-based Watermarking. \textbf{Top}: The watermark signal is embedded in the noise vector in the latent space, a watermarked table is subsequently generated. \textbf{Bottom}: To detect the watermark signal, we need to reverse the entire pipeline. IQT stands for the inverse map of Quantile Transformation.}
  \label{fig:tabsyn_architecture}
\end{figure}

\subsection{Inversion of (Inverse) Quantile Transformation}
\label{appendix:iq_inversion}
The Quantile Transformation~\citep{wikipedia_quantile_normalization} is a widely used~\citep{tabsyn,tabdar,tabdiff,tabddpm} data preprocessing step in tabular data synthesis. It regularizes the data distribution to a standard normal distribution. The Quantile Transformation can be implemented as follows:
\begin{enumerate}
    \item Estimate the empirical cumulative distribution function (CDF) of the features.
    \item Map to uniform distribution with the estimated CDF.
    \item Map to standard normal distribution with inverse transform sampling: $z = \Phi^{-1}(u)$, where $\Phi$ is the CDF of the standard normal distribution.
\end{enumerate}
Note that in the second step, only the ordering of the data is preserved, and the exact values are not preserved, making the map non-injective, therefore, the inverse of the Quantile Transformation is inherently error-prone. 
Based on the official codebase, TabWak~\citep{tabwak} bypass the inversion of quantile normalization by caching the original data during watermarking, this is infeasible in practical scenarios where the ground truth is unavailable. To study the impact of the inversion error of the Quantile Transformation, we apply the original Quantile Transformation to the sampled tabular data to inverse the inverse quantile transformation.

\subsection{Inversion of VAE decoder}
\label{appendix:vae_inversion}
Denote the VAE decoder as $f_{\theta}$, and the VAE decoder output as $\rvx = f_{\theta}(\rvz)$. To get $\rvz$ from $\rvx$, \citep{tabwak} employ a gradient-based optimization to approximate the inverse of the VAE decoder. Specifically, we can parametrize the unknown $\rvz$ with trainable parameters, and optimize the following objective with standard gradient descent:
\begin{equation} \nonumber
  \label{eq:vae_decoder_inverse}
  \rvz = \arg\min_{\rvz} \left\| \rvx - f_{\theta}(\rvz) \right\|_2^2.
\end{equation}
where $\rvz$ is inilitaized as $g(f_{\theta}(\rvx))$, and $g(\cdot)$ is a VAE encoder. However, there is no guarantee that the above optimization will converge to the true $\rvz$, and we observed that the optimization process is unstable (sometimes produce $\text{NaN}$) for tabular data and introduce significant error in the inversion process.

\subsection{DDIM Inversion} 
\label{appendix:ddim_inversion}
The DDIM diffusion forward process is defined as:
\begin{equation} \nonumber
q(\rvx_t \mid \rvx_{t-1}) = \mathcal{N}(\rvx_t; \sqrt{1-\beta_t} \rvx_{t-1}, \beta_t \mathbf{I}),
\end{equation}
where $\rvx_0$ is the original data, $\rvx_t$ is the data at time $t$, and $\beta_t$ is the variance of the noise at step $t$.
Based on the above definition, we can write $\rvx_t$ as:
\begin{equation} \tag{Forward process}
\rvx_t = \sqrt{\bar{\alpha}_t} \rvx_{t-1} + \sqrt{1-\bar{\alpha}_t} \epsilon,
\end{equation}
where $\bar{\alpha}_t = \prod_{i=0}^t (1-\beta_i)$, $\epsilon \sim \mathcal{N}(\mathbf{0}, \mathbf{I})$.

Starting from $\rvx_T$, we sample $\rvx_{T-1}, \ldots, \rvx_0$ recursively according to the following process:
\begin{equation} \tag{Reverse process}
\label{eq:reverse_process}
\begin{aligned}
  &\rvx_0^t = \left(\rvx_t-\sqrt{1-\bar{\alpha}_t} \epsilon_\theta(\rvx_t, t)\right) / \sqrt{\bar{\alpha}_t} \\
  &\rvx_{t-1} = \sqrt{\bar{\alpha}_{t-1}} \rvx_0^t+\sqrt{1-\bar{\alpha}_{t-1}} \epsilon_\theta(\rvx_t, t),
\end{aligned}
\end{equation}
where $\epsilon_\theta(\rvx_t, t)$ is noise predicted by a neural network.

The \textbf{DDIM inversion process} is defined as the inverse of the DDIM reverse process. Specifically, starting from $\rvx_0$, our goal is to recover the original noise vector $\rvx_T$ in the latent space.
We introduce the basic DDIM inversion process proposed in \citep{dhariwal2021diffusion}, and is widely adapted in inversion-based watermark methods~\citep{tree_ring_watermark,gaussian_shading,tabwak,hu2025videoshield}.

We can abtain the inverse of the DDIM forward process by replacing the $t-1$ subscript with $t+1$ in \Cref{eq:reverse_process}, but use $\rvx_t$ to approxiate the unknown $\rvx_{t+2}$:
\begin{equation} \nonumber
  \label{eq:reverse_process_inverse}
  \begin{aligned}
    &\rvx_{t+1} = \sqrt{\bar{\alpha}_{t+1}} \rvx_0^t+\sqrt{1-\bar{\alpha}_{t+1}} \epsilon_\theta(\rvx_{t}, t),
  \end{aligned}
\end{equation}
Due to the approximation $\rvx_t \approx \rvx_{t+2}$, the inversion process generally demands a finer discretization of the time steps. For instance, inversion-based watermarking methods~\citep{tree_ring_watermark,tabwak} typically adopt $T=1000$ steps, whereas diffusion models optimized for fast inference~\citep{edm,tabsyn} often operate with a coarser discretization of $T=50$ steps. 
\paragraph{Advanced Inversion Methods.} To address the inexactness of the above inversion process, recent works \citep{hong2024exact,pan2023effective} have proposed more accurate inversion methods based on iterative optimization. However, we empirically found that those methods still suffer from inversion error due to already noisy input from the previous steps (VAE decoder and Quantile Transformation).

\begin{figure}[H]
    \centering
    \includegraphics[width=0.8\textwidth]{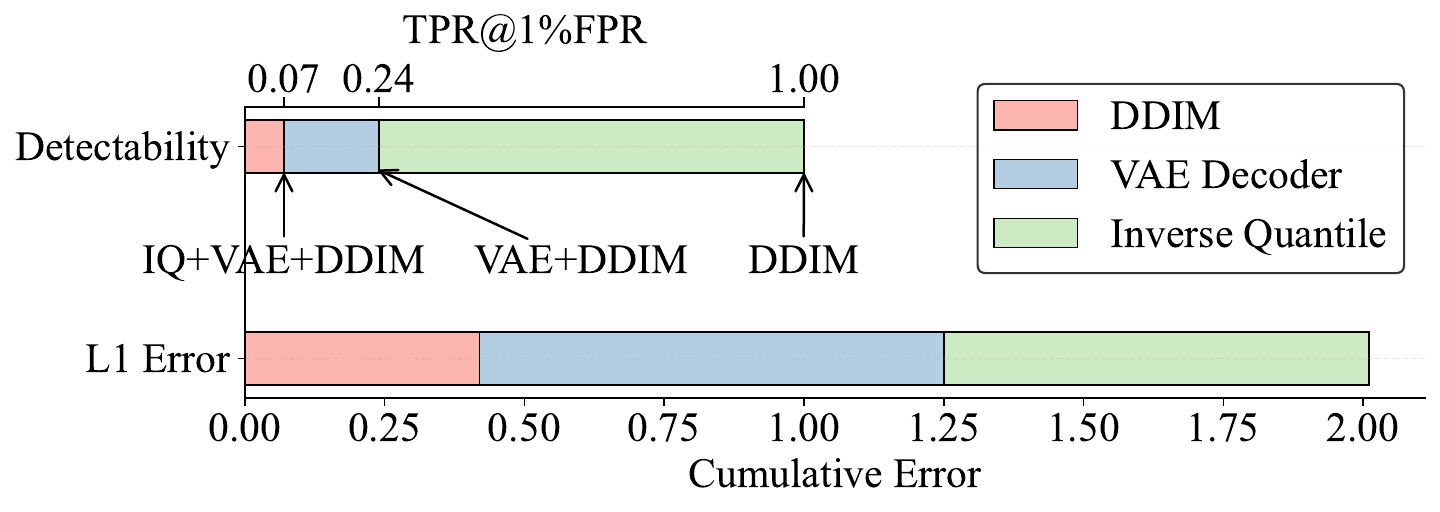}
    \caption{Error Accumulation and Detection Performance Across Inversion Stages of TabWak. The $\ell_1$ error is computed between the estimated and ground truth noise vectors in latent space.}
    \label{fig:error_vs_perf}
\end{figure}

\subsection{Error Accumulation}
\label{appendix:error_accumulation}
In \Cref{fig:error_vs_perf}, we analyze the error accumulated at each inversion stage and its impact on detection performance using the \texttt{Adult} dataset. Specifically, we compute the TPR@1\%FPR over 100 watermarked tables, each with 100 rows. The top bar chart shows detection performance when progressively inverting different parts of the pipeline. From left to right:
\begin{itemize}
    \item When we invert the entire pipeline (IQ $\rightarrow$ VAE $\rightarrow$ DDIM), the detection performance drops to 0.07 TPR@1\%FPR.
    \item When we provide the ground-truth IQ and only invert the VAE decoder and DDIM, the performance improves to 0.24 TPR@1\%FPR.
    \item When both the ground-truth IQ and VAE decoder outputs are provided (i.e., only DDIM is inverted), detection reaches a perfect 1.0 TPR@1\%FPR.
\end{itemize}

The bottom bar chart reports the $\ell_1$ error between the estimated and ground-truth noise vectors in the latent space. From left to right, the bars correspond to:
\begin{itemize}
    \item Inverting only DDIM (given the ground-truth VAE output),
    \item Inverting both the VAE decoder and DDIM (given the ground-truth IQ), and
    \item Inverting the full pipeline (IQ $\rightarrow$ VAE $\rightarrow$ DDIM).
\end{itemize}

This comparison highlights how errors accumulate through the inversion stages and directly affect watermark detectability.

\section{Experimental Details}
\subsection{Hardware Specification}
We use a single hardware for all experiments. The hardware specifications are as follows:
\begin{itemize}
    \item GPU: NVIDIA RTX 4090
    \item CPU: Intel 14900K
\end{itemize}

\subsection{Dataset Statistics}\label{appendix:dataset}
The dataset used in this paper could be automatically downloaded using the script in the provided code.
We use $6$ tabular datasets from UCI Machine Learning Repository\footnote{\url{https://archive.ics.uci.edu/datasets}}: Adult\footnote{\url{https://archive.ics.uci.edu/dataset/2/adult}}, Default\footnote{\url{https://archive.ics.uci.edu/dataset/350/default+of+credit+card+clients}}, Shoppers\footnote{\url{https://archive.ics.uci.edu/dataset/468/online+shoppers+purchasing+intention+dataset}}, Magic\footnote{\url{https://archive.ics.uci.edu/dataset/159/magic+gamma+telescope}}, Beijing\footnote{\url{https://archive.ics.uci.edu/dataset/381/beijing+pm2+5+data}}, and News\footnote{\url{https://archive.ics.uci.edu/dataset/332/online+news+popularity}}, which contains varies number of numerical and categorical features. The statistics of the datasets are presented in Table~\ref{tbl:exp-dataset}.
\begin{table}[!h] 
    \vspace{-0.5cm}
    \centering
    \caption{Dataset statistics.} 
    \label{tbl:exp-dataset}
    \small
    \begin{threeparttable}
    {
    \resizebox{\columnwidth}{!}
    {
	\begin{tabular}{lcccccccc}
            \toprule[0.8pt]
            \textbf{Dataset} & \# Rows  & \# Continuous & \# Discrete  &  \# Target & \# Train &  \# Test & Task  \\
            \midrule 
            \textbf{Adult} & $32,561$ & $6$ & $8$ & $1$ & $22,792$ & $16,281$ & Classification  \\
            \textbf{Default} & $30,000$ & $14$ & $10$ & $1$ & $27,000$ & $3,000$  & Classification   \\
            \textbf{Shoppers} & $12,330$ & $10$ & $7$ & $1$ & $11,098$ & $1,232$ & Classification   \\
            \textbf{Magic} & $19,021$ & $10$ & $1$ & $1$ & $17,118$ & $1,903$ & Classification  \\
            \textbf{Beijing} & $43,824$ & $7$ & $5$  & $1$& $39,441$ & $4,383$ &  Regression   \\
            \textbf{News} & $39,644$ & $46$ & $2$  & $1$& $35,679$ & $3,965$ & Regression \\
		\bottomrule[1.0pt] 
		\end{tabular}
   }
  }        
  \end{threeparttable}
\end{table}

In Table~\ref{tbl:exp-dataset}, \textbf{\# Rows} refers to the total records in each dataset, while \textbf{\# Continuous} and \textbf{\# Discrete} denote the count of numerical and categorical features, respectively. The \textbf{\# Target} column indicates whether the prediction task involves a continuous (regression) or discrete (classification) target variable. All datasets except Adult are partitioned into training and testing sets using a 9:1 ratio, with splits generated using a fixed random seed for reproducibility. The Adult dataset uses its predefined official testing set. For evaluating Machine Learning Efficiency (MLE), the training data is further subdivided into training and validation subsets with an 8:1 ratio, ensuring consistent evaluation protocols across experiments.

\subsection{Fidelity Metrics}\label{appendix:metric}
The fidelity metrics used in this paper (Marginal, Correlation, C2ST and MLE) are standard metrics in the field of tabualr data synthesis. Here is a reference:
\begin{itemize}
    \item Marginal: Appendix E.3.1 in~\citep{tabsyn}.
    \item Correlation: Appendix E.3.2 in~\citep{tabsyn}.
    \item C2ST: Appendix F.3 in~\citep{tabsyn}.
    \item MLE:  Appendix E.4 in~\citep{tabsyn}.
\end{itemize}

Below is a summary of how these metrics work.

\subsubsection{Marginal Distribution}
The \textbf{Marginal} metric assesses how well the marginal distribution of each column is preserved in the synthetic data. For continuous columns, we use the Kolmogorov–Smirnov Test (KST); for categorical columns, we use the Total Variation Distance (TVD).

\paragraph{Kolmogorov–Smirnov Test (KST)}  
Given two continuous distributions $p_r(x)$ and $p_s(x)$ (real and synthetic, respectively), the KST measures the maximum discrepancy between their cumulative distribution functions (CDFs):
\begin{equation}
    \mathrm{KST} = \sup_{x} \left| F_r(x) - F_s(x) \right|,
\end{equation}
where $F_r(x)$ and $F_s(x)$ denote the CDFs of $p_r(x)$ and $p_s(x)$:
\begin{equation}
    F(x) = \int_{-\infty}^x p(x) \, \mathrm{d}x.
\end{equation}

\paragraph{Total Variation Distance (TVD)}  
TVD measures the difference between the categorical distributions of real and synthetic data. Let $\Omega$ be the set of possible categories in a column. Then:
\begin{equation}
    \mathrm{TVD} = \frac{1}{2} \sum_{\omega \in \Omega} \left| R(\omega) - S(\omega) \right|,
\end{equation}
where $R(\cdot)$ and $S(\cdot)$ denote the empirical probabilities in real and synthetic data, respectively.

\vspace{0.5em}
\subsubsection{Correlation}
The \textbf{Correlation} metric evaluates whether pairwise relationships between columns are preserved.

\paragraph{Pearson Correlation Coefficient}  
For two continuous columns $x$ and $y$, the Pearson correlation coefficient is defined as:
\begin{equation}
    \rho_{x, y} = \frac{\mathrm{Cov}(x, y)}{\sigma_x \sigma_y},
\end{equation}
where $\mathrm{Cov}(\cdot)$ is the covariance and $\sigma$ denotes standard deviation.  
We evaluate the preservation of correlation by computing the mean absolute difference between correlations in real and synthetic data:
\begin{equation}
    \text{Pearson Score} = \frac{1}{2} \mathbb{E}_{x,y} \left| \rho^R(x, y) - \rho^S(x, y) \right|,
\end{equation}
where $\rho^R$ and $\rho^S$ denote correlations in real and synthetic data. The score is scaled by $\frac{1}{2}$ to ensure it lies in $[0,1]$. Lower values indicate better alignment.

\paragraph{Contingency Similarity}  
For categorical columns $A$ and $B$, we compute the Total Variation Distance between their contingency tables:
\begin{equation}
    \text{Contingency Score} = \frac{1}{2} \sum_{\alpha \in A} \sum_{\beta \in B} \left| R_{\alpha,\beta} - S_{\alpha,\beta} \right|,
\end{equation}
where $R_{\alpha,\beta}$ and $S_{\alpha,\beta}$ are the joint frequencies of $(\alpha, \beta)$ in the real and synthetic data, respectively.

\vspace{0.5em}
\subsubsection{Classifier Two-Sample Test (C2ST)}
C2ST evaluates how distinguishable the synthetic data is from real data. If a classifier can easily separate the two, the synthetic data poorly approximates the real distribution. We adopt the implementation provided by the SDMetrics library.\footnote{\url{https://docs.sdv.dev/sdmetrics/metrics/metrics-in-beta/detection-single-table}}

\vspace{0.5em}
\subsubsection{Machine Learning Efficiency (MLE)}
MLE evaluates the utility of synthetic data for downstream machine learning tasks. Each dataset is split into training and testing subsets using real data. Generative models are trained on the real training set, and a synthetic dataset of equal size is sampled.

For both real and synthetic data, we use the following protocol:
\begin{itemize}
  \item Split the training set into train/validation with an $8{:}1$ ratio.
  \item Train a classifier/regressor on the train split.
  \item Tune hyperparameters based on validation performance.
  \item Retrain the model on the full training set using the optimal hyperparameters.
  \item Evaluate on the real test set.
\end{itemize}

This process is repeated over $20$ random train/validation splits. Final scores (AUC for classification task or RMSE for regression task) are averaged over the $20$ trials for both real and synthetic training data.
In our experiments, we report the MLE Gap which is the difference between the MLE score of the (unwatermarked) real data and the MLE score of the synthetic data.

\subsection{Watermark Detection Metrics}
For watermark detection metrics, we primaryly use the area under the curve (AUC) of the receiver operating characteristic (ROC) curve: \textbf{AUC}, and the True Positive Rate (TPR) at a given False Positive Rate (FPR): \textbf{TPR@x$\%$FPR}.

\paragraph{$z$-statistic}
In addition, we can formalize a statistical test for watermark detection for \modelname. Specifically, consider a table $T$ containing $N$ samples (rows) $\rvx_1, \ldots, \rvx_N$. Recall that during watermarking, each row is assigned a binary score of $0$ or $1$ based on a pseudorandom function, and the row that scores higher is kept. Therefore, for a watermarked table, the total count of rows with a score $1$, denoted by $|W|$, is expected to be significantly higher than random chance. To statistically validate this, we formulate watermark detection as a hypothesis testing problem:
\begin{equation}\nonumber
    \begin{aligned}
        H_0: \text{The table is generated without watermarking.} \\
        \text{vs. } H_1: \text{The table is generated with watermarking.}
    \end{aligned}
\end{equation}
Under the null hypothesis, $|W|$ follows a binomial distribution with mean $\mu = N/2$ and variance $\sigma^2 = N/4$. The standardized $z$-statistic is computed as: 
\begin{equation}\nonumber
    z = \frac{|W| -\mu}{\sigma} = \frac{|W| - N/2}{\sqrt{N/4}}.
\end{equation}
We perform a \textit{one-tailed} test (upper tail) since the alternative hypothesis predicts $|W|>N/2$. The $z$-statistic is compared against a critical value $z_{\alpha}$ corresponding to a desired significance level $\alpha$ (e.g. $\alpha=0.05$ yields $z_{\alpha}=1.645$). If $z>z_{\alpha}$, we reject the null hypothesis and conclude that the table is watermarked.

\section{Ommited Proofs in \Cref{sec:methods}}
Recall that for a table $T$ (wateramarked or unwatermarked) with $N$ rows: $\rvx_1, \dots, \rvx_N$, we define the watermark detection score as 
\begin{equation}
  \label{eq:watermark_detection_score}
    S(T) = \frac{1}{N} \sum_{i=1}^N s_k(\rvx_i),
\end{equation}
where $s_k(\rvx_i)$ is the score of the $i$-th sample, $k$ is the fixed watermark key.

\falsePositiveRateBound*
\begin{proof}
Let $S(T_{\mathrm{no\text{-}wm}}) = \sum_{i=1}^N \rc_i$ denote the sum of $N$ i.i.d. scores from the unwatermarked table, where each $\rc_i = s_k(\rvx_i)$ for $\rvx_i \sim p(\rvx)$, and similarly let $S(T_{\mathrm{wm}}) = \sum_{i=1}^N \rc'_i$ denote the sum of $N$ i.i.d. scores from the watermarked table, where each $\rc'_i = \max\{s_k(\rvx_{i1}), \dots, s_k(\rvx_{im})\}$ with $\rvx_{ij} \sim p(\rvx)$.

Define the expected values:
\[
\mu_{\mathrm{no\text{-}wm}} = \mathbb{E}[\rc_i], \quad
\mu_{\mathrm{wm}}^m = \mathbb{E}[\rc'_i].
\]

We are interested in bounding the false positive rate:
\[
\Pr(S(T_{\mathrm{no\text{-}wm}}) > S(T_{\mathrm{wm}})) = \Pr\left( \sum_{i=1}^N (\rc_i - \rc'_i) > 0 \right).
\]

Let $\rw_i = \rc_i - \rc'_i$. Since $s_k(x) \in [0,1]$, we have $\rc_i \in [0,1]$ and $\rc'_i \in [0,1]$, so $\rw_i \in [-1,1]$. Moreover, $\mathbb{E}[\rw_i] = \mu_{\mathrm{no\text{-}wm}} - \mu_{\mathrm{wm}}^m =: -\delta$, where $\delta = \mu_{\mathrm{wm}}^m - \mu_{\mathrm{no\text{-}wm}} > 0$.

We apply Hoeffding's inequality to the sum of $\rw_i$'s:
\[
\Pr\left( \sum_{i=1}^N \rw_i > 0 \right) = \Pr\left( \sum_{i=1}^N \rw_i - \mathbb{E}[\sum_{i=1}^N \rw_i] > N \delta \right)
\le \exp\left( -\frac{2 N^2 \delta^2}{4N} \right).
\]

Plug in the definition of $\delta$, we have:
\[
\Pr(S(T_{\mathrm{no\text{-}wm}}) > S(T_{\mathrm{wm}})) \le \exp\left( -\frac{N^2 \delta^2}{2} \right)
= \exp\left( -\frac{N (\mu_{\mathrm{wm}}^m - \mu_{\mathrm{no\text{-}wm}})^2}{2} \right).
\]
which proves the result.
\end{proof}

\optimalScoringDistribution*
\begin{proof}
  Let $s_1, \dots, s_m$ be $i.i.d.$ copies of a random variable $s_k(\rvx) \in [0,1]$ with fixed mean $\mathbb{E}[s_k(\rvx)] = 0.5$. Define:
  \[
  \mu := \mathbb{E}[s_k(\rvx)] = 0.5, \quad \mu_{\max} := \mathbb{E}[\max(s_1, \dots, s_m)].
  \]
  Let $\Delta := \mu_{\max} - \mu$ be the gap between the expected maximum score over $m$ repetitions and the mean score. The upper bound in \Cref{eq:fpr_bound} is:
  \[
  \Pr(S_{\text{no-wm}} > S_{\text{wm}}) \le \exp\left( -\frac{N \Delta^2}{2} \right),
  \]
  so minimizing the FPR corresponds to maximizing $\Delta$ under the constraint that $\mathbb{E}[s_k(\rvx)] = 0.5$ and $s_k(\rvx) \in [0,1]$.
  
  We now show that $\Delta$ is maximized when $s_k(\rvx) \sim \text{Bernoulli}(0.5)$.
  
  \paragraph{Step 1: Write $\mu_{\max}$ and $\mu$ as integrals over the CDF.}
  Let $F$ be the cumulative distribution function (CDF) of $s_k(\rvx)$. Then the CDF of $\max(s_1, \dots, s_m)$ is $F^m(x)$. 
  By the tail integration formula, we can compute the expected maximum as:
  \begin{align*}
  \mu_{\max} 
  &= \int_0^1 \mathrm{Pr}(\max(s_1, \dots, s_m) > x) \,  \\
  &= \int_0^1 (1 - F(x)^m) \, dx.
  \end{align*}
  
  Similarly, we have: $\mu = \int_0^1 (1 - F(x)) \, dx$.

  Therefore, the gap $\Delta$ can be written as:
  \[
  \Delta = \mu_{\max} - \mu = \int_0^1 \left[F(x) - F(x)^m\right] dx.
  \]
  
  \paragraph{Step 2: Leverage the concavity.}
  By \Cref{lem:concavity}, the integrand $F(x) - F(x)^m$ is concave in $F(x)$. By \Cref{lem:extremal_support}, the integral is maximized when $F(x)$ is the CDF of a Bernoulli distribution with mean $\mu = 0.5$.

  Therefore, among all $s_k(\rvx) \in [0,1]$ with $\mathbb{E}[s_k(\rvx)] = 0.5$, the Bernoulli(0.5) distribution maximizes $\Delta$, which minimizes the upper bound on the FPR. Hence, the lemma holds.
  \end{proof}

\minimumWatermarkingSignal*
\begin{proof}
  When $s_k(\rvx) \sim \text{Bernoulli}(0.5)$, we have:
  \[
  \mu_{\mathrm{no\text{-}wm}} = \mathbb{E}[s_k(\rvx)] = 0.5,
  \quad \mu_{\mathrm{wm}}^m = \mathbb{E}[\max(s_1, \dots, s_m)] = 1 - 0.5^m.
  \]
  Plug in into the FPR bound \Cref{eq:upper_bound}, we have:
  \[
  \Pr\left(S(T_{\mathrm{no\text{-}wm}}) > S(T_{\mathrm{wm}})\right)
  \le \exp\left( -\frac{N}{2} \left(0.5 - 0.5^m\right)^2 \right),
  \]
  which completes the proof.
\end{proof}

\section{Technical Lemmas}
\begin{lemma}
  \label{lem:concavity}
  For any integer $m \ge 2$, the function $f(x) = x - x^m$ is concave on the interval $[0,1]$.
\end{lemma}
\begin{proof}
  To prove that $f(x) = x - x^m$ is concave on $[0,1]$, we show that its second derivative is non-positive on this interval.
  
  Compute the first derivative:
  \[
  f'(x) = \frac{d}{dx}(x - x^m) = 1 - m x^{m-1}.
  \]
  
  Compute the second derivative:
  \[
  f''(x) = \frac{d}{dx}(1 - m x^{m-1}) = -m(m-1) x^{m-2}.
  \]
  
  Observe that for all $x \in [0,1]$ and $m \ge 2$: $m(m-1) > 0$ and $x^{m-2} \ge 0$.
  
  Therefore,
  \[
  f''(x) = -m(m-1)x^{m-2} \le 0 \quad \text{for all } x \in [0,1].
  \]
  
  Hence, $f(x)$ is concave on $[0,1]$.
\end{proof}
  
\begin{lemma}
  \label{lem:extremal_support}
  Let $\phi : [0,1] \to \mathbb{R}$ be a concave function, and let $F$ be the cumulative distribution function (CDF) of a random variable supported on $[0,1]$ with fixed mean $\mu \in (0,1)$.
  Then the integral
  \[
  \int_0^1 \phi(F(x)) dx
  \]
  is maximized when $F(x)= \begin{cases} 0 & \text{if } x < 0 \\ 1-\mu & \text{if } 0 \le x < 1 \\ 1 & \text{if } x \ge 1 \end{cases}$, i.e. the CDF of a Bernoulli distribution with mean $\mu$.
\end{lemma}
  
\begin{proof}
  \textbf{Step 1: Rewrite the Mean Constraint}
  
  By the tail integration formula, the mean constraint for the random variable $X$ with CDF $F(x)$ supported on $[0,1]$ is:
  $$ \int_0^1 (1 - F(x)) \,dx = \mu. $$
  Rearranging this equation gives the integral of $F(x)$:
  \begin{equation} 
  \label{eq:mean_constraint}
  \int_0^1 F(x) \,dx = 1 - \mu.
  \end{equation}
  
  \textbf{Step 2: Upper Bound the Integral}
  
  The function $\phi : [0,1] \to \mathbb{R}$ is concave. The CDF $F(x)$ takes values in $[0,1]$ for $x \in [0,1]$, so $\phi(F(x))$ is well-defined. We can apply Jensen's inequality for integrals, which for a concave function $\phi$ and an integrable function $g(x)$ on an interval $[a,b]$ states:
  $$ \frac{1}{b-a} \int_a^b \phi(g(x)) \,dx \le \phi\left(\frac{1}{b-a} \int_a^b g(x) \,dx\right). $$
  Plug in $a=0$, $b=1$, $g(x) = F(x)$. Jensen's inequality then becomes:
  $$ \int_0^1 \phi(F(x)) \,dx \le \phi\left(\int_0^1 F(x) \,dx\right). $$
  Substituting \Cref{eq:mean_constraint} into the right hand side, we have:
  \begin{equation}
  \label{eq:upper_bound}
  \int_0^1 \phi(F(x)) \,dx \le \phi(1-\mu).
  \end{equation}
  
\textbf{Step 3: Verify $F(x)$ achieves the upper bound}

It is straightforward to verify that $F(x)$ satisfies the mean constraint. Next, we will show that $F(x)$ achieves the upper bound $\phi(1-\mu)$.
For $x \in [0,1)$, $F(x) = 1-\mu$. Therefore, we have:
$$ \int_0^1 \phi(F(x)) \,dx = \int_0^1 \phi(1-\mu) \,dx = \phi(1-\mu). $$

We have shown that $F(x)$ satisfies the mean constraint and achieves the upper bound $\phi(1-\mu)$, which completes the proof.
\end{proof}

%% file: Tables/quality_detect_appendix.tex
\begin{table}[H]
    \vspace{-0.5cm}
    \caption{Watermark generation quality and detectability, \raisebox{0.5ex}{\colorbox{blue!30}{\quad}} indicates best performance, \raisebox{0.5ex}{\colorbox{blue!10}{\quad}} indicates second-best performance. For clarity, only our method is highlighted in detection.}
    \vspace{3pt}
    \centering
    \resizebox{0.98\textwidth}{!}{
    \begin{threeparttable}
    \begin{tabular}{llcccc|cccc}
    \toprule[1.0pt]
    & &\multicolumn{4}{c}{Watermark Generation Quality} & \multicolumn{4}{c}{Watermark Detectability}\\
    \cmidrule(r){3-6} \cmidrule(r){7-10}
    Dataset & Method & \multicolumn{4}{c}{Num. Training Rows} &  \multicolumn{2}{c}{100} & \multicolumn{2}{c}{500} \\
    \cmidrule(r){3-6} \cmidrule(r){7-8} \cmidrule(r){9-10}
     & & \textbf{Marg.} & \textbf{Corr.} & \textbf{C2ST} & \textbf{MLE Gap} & \textbf{AUC} & \textbf{T@0.1\%F} & \textbf{AUC} & \textbf{T@0.1\%F} \\
    \midrule
    \midrule
    \multirow{7}{*}{Beijing} & w/o WM & 0.977 & 0.958 & 0.934 & 0.199 & - & - & - & - \\ 
    & TR & 0.914 & \cellcolor{blue!10}0.873 & 0.734 & 0.396 & 0.577 & 0.000 & 0.548 & 0.007 \\
    & GS & 0.656 & 0.529 & 0.097 & 0.715 & 1.000 & 1.000 & 1.000 & 1.000 \\
    & TabWak & \cellcolor{blue!10}0.923 & 0.871 & \cellcolor{blue!10}0.792 & \cellcolor{blue!10}0.375 & 0.925 & 0.096 & 0.999 & 0.978 \\
    & TabWak* & 0.917 & 0.860 & 0.761 & 0.403 & 0.996 & 0.734 & 1.000 & 1.000 \\
    & \textbf{\modelname} & \cellcolor{blue!30}0.972 & \cellcolor{blue!30}0.955 & \cellcolor{blue!30}0.926 & \cellcolor{blue!30}0.209 & \cellcolor{blue!30}1.000 & \cellcolor{blue!30}1.000 & \cellcolor{blue!30}1.000 & \cellcolor{blue!30}1.000 \\
    \midrule
    \multirow{7}{*}{News} & w/o WM & 0.960 & 0.973 & 0.899 & 0.024 & - & - & - & - \\ 
    & TR & 0.899 & 0.963 & 0.641 & \cellcolor{blue!10}0.041 & 0.547 & 0.000 & 0.549 & 0.005 \\
    & GS & 0.673 & 0.907 & 0.031 & 0.065 & 1.000 & 1.000 & 1.000 & 1.000 \\
    & TabWak & \cellcolor{blue!10}0.929 & \cellcolor{blue!10}0.968 & \cellcolor{blue!10}0.749 & 0.066 & 0.998 & 0.869 & 1.000 & 1.000 \\
    & TabWak* & 0.924 & 0.964 & 0.719 & 0.044 & 1.000 & 0.991 & 1.000 & 1.000 \\
    & \textbf{\modelname} & \cellcolor{blue!30}0.959 & \cellcolor{blue!30}0.973 & \cellcolor{blue!30}0.883 & \cellcolor{blue!30}0.033 & \cellcolor{blue!30}1.000 & \cellcolor{blue!30}1.000 & \cellcolor{blue!30}1.000 & \cellcolor{blue!30}1.000 \\
    \bottomrule[1.0pt]
    \end{tabular}
    \end{threeparttable}
    }
    \vspace{-0.3cm}
    \label{tbl:quality_detect_appendix}
\end{table}

%% file: Tables/edit_baseline.tex
\begin{table}[t]
    \caption{Watermark generation quality and detectability, \raisebox{0.5ex}{\colorbox{blue!30}{\quad}} indicates best performance, \raisebox{0.5ex}{\colorbox{blue!10}{\quad}} indicates second-best performance. For clarity, only our method is highlighted in detection.}
    \vspace{8pt}
    \centering
    \resizebox{0.98\textwidth}{!}{
    \begin{threeparttable}
    \begin{tabular}{llcccc|cccc}
    \toprule[1.0pt]
    & &\multicolumn{4}{c}{Watermark Generation Quality} & \multicolumn{4}{c}{Watermark Detectability}\\
    \cmidrule(r){3-6} \cmidrule(r){7-10}
    Dataset & Method & \multicolumn{4}{c}{Num. Training Rows} &  \multicolumn{2}{c}{100} & \multicolumn{2}{c}{500} \\
    \cmidrule(r){3-6} \cmidrule(r){7-8} \cmidrule(r){9-10}
     & & \textbf{Marg.} & \textbf{Corr.} & \textbf{C2ST} & \textbf{MLE Gap} & \textbf{AUC} & \textbf{T@0.1\%F} & \textbf{AUC} & \textbf{T@0.1\%F} \\
    \midrule
    \multirow{5}{*}{Adult} & w/o WM & 0.994 & 0.984 & 0.996 & 0.017 & - & - & - & - \\ 
    & TabularMark & \cellcolor{blue!10}0.983 & 0.949 & \cellcolor{blue!30}0.987 & 0.021 & 1.000 & 1.000 & 1.000 & 1.000 \\
    & WGTD & \cellcolor{blue!30}0.987 & \cellcolor{blue!30}0.972 & \cellcolor{blue!10}0.978 & \cellcolor{blue!10}0.019 & 1.000 & 1.000 & 1.000 & 1.000 \\
    & \textbf{\modelname} & 0.979 & \cellcolor{blue!10}0.963 & 0.883 & \cellcolor{blue!30}0.017 & \cellcolor{blue!30}1.000 & \cellcolor{blue!30}1.000 & \cellcolor{blue!30}1.000 & \cellcolor{blue!30}1.000 \\
    \midrule
    \multirow{5}{*}{Beijing} & w/o WM & 0.977 & 0.958 & 0.934 & 0.199 & - & - & - & - \\ 
    & TabularMark & 0.935 & 0.789 & \cellcolor{blue!30}0.941 & 0.528 & 1.000 & 1.000 & 1.000 & 1.000 \\
    & WGTD & \cellcolor{blue!10}0.964 & \cellcolor{blue!10}0.948 & \cellcolor{blue!10}0.929 & \cellcolor{blue!10}0.527 & 1.000 & 1.000 & 1.000 & 1.000 \\
    & \textbf{\modelname} & \cellcolor{blue!30}0.972 & \cellcolor{blue!30}0.955 & 0.926 & \cellcolor{blue!30}0.209 & \cellcolor{blue!30}1.000 & \cellcolor{blue!30}1.000 & \cellcolor{blue!30}1.000 & \cellcolor{blue!30}1.000 \\
    \midrule
    \multirow{5}{*}{Default} & w/o WM & 0.990 & 0.934 & 0.979 & 0.000 & - & - & - & - \\ 
    & TabularMark & \cellcolor{blue!10}0.987 & 0.939 & \cellcolor{blue!10}0.961 & 0.004 & 1.000 & 1.000 & 1.000 & 1.000 \\
    & WGTD & \cellcolor{blue!30}0.989 & \cellcolor{blue!30}0.913 & \cellcolor{blue!30}0.919 & \cellcolor{blue!30}0.000 & 1.000 & 1.000 & 1.000 & 1.000 \\
    & \textbf{\modelname} & 0.983 & \cellcolor{blue!10}0.925 & 0.963 & \cellcolor{blue!10}0.002 & \cellcolor{blue!30}1.000 & \cellcolor{blue!30}1.000 & \cellcolor{blue!30}1.000 & \cellcolor{blue!30}1.000 \\
    \midrule
    \multirow{5}{*}{Magic} & w/o WM & 0.990 & 0.980 & 0.998 & 0.008 & - & - & - & - \\
    & TabularMark & \cellcolor{blue!10}0.985 & 0.975 & \cellcolor{blue!30}0.999 & 0.026 & 1.000 & 1.000 & 1.000 & 1.000 \\
    & WGTD & 0.979 & \cellcolor{blue!10}0.977 & 0.998 & \cellcolor{blue!10}0.019 & 1.000 & 1.000 & 1.000 & 1.000 \\
    & \textbf{\modelname} & \cellcolor{blue!30}0.991 & \cellcolor{blue!30}0.982 & \cellcolor{blue!30}0.999 & \cellcolor{blue!30}0.010 & \cellcolor{blue!30}1.000 & \cellcolor{blue!30}1.000 & \cellcolor{blue!30}1.000 & \cellcolor{blue!30}1.000 \\
    \midrule    
    \multirow{4}{*}{News} & w/o WM & 0.960 & 0.973 & 0.811 & 0.024 & - & - & - & - \\ 
    & TabularMark & \cellcolor{blue!30}0.959 & \cellcolor{blue!10}0.969 & \cellcolor{blue!10}0.877 & \cellcolor{blue!10}0.130 & 1.000 & 1.000 & 1.000 & 1.000 \\
    & WGTD & 0.903 & 0.968 & 0.861 & 0.131 & 1.000 & 1.000 & 1.000 & 1.000 \\
    & \textbf{\modelname} & \cellcolor{blue!30}0.959 & \cellcolor{blue!30}0.973 & \cellcolor{blue!30}0.883 & \cellcolor{blue!30}0.033 & \cellcolor{blue!30}1.000 & \cellcolor{blue!30}1.000 & \cellcolor{blue!30}1.000 & \cellcolor{blue!30}1.000 \\
    \midrule
    \multirow{4}{*}{Shoppers} & w/o WM & 0.985 & 0.974 & 0.974 & 0.017 & - & - & - & - \\
    & TabularMark & \cellcolor{blue!10}0.974 & 0.930 & \cellcolor{blue!30}0.975 & \cellcolor{blue!30}0.013 & 1.000 & 1.000 & 1.000 & 1.000 \\
    & WGTD & 0.964 & \cellcolor{blue!10}0.944 & 0.887 & 0.008 & 1.000 & 1.000 & 1.000 & 1.000 \\
    & \textbf{\modelname} & \cellcolor{blue!30}0.982 & \cellcolor{blue!30}0.974 & \cellcolor{blue!10}0.950 & \cellcolor{blue!10}0.015 & \cellcolor{blue!30}1.000 & \cellcolor{blue!30}1.000 & \cellcolor{blue!30}1.000 & \cellcolor{blue!30}1.000 \\
    \bottomrule[1.0pt]
    \end{tabular}
    \end{threeparttable}
    }
    \label{tbl:edit_baseline}
\end{table}